\documentclass[12pt, english]{article}
\usepackage{mathrsfs}
\usepackage{amssymb}
\usepackage{amsfonts}
\usepackage{amsthm}
\usepackage{amsmath}
\usepackage{lscape}
\usepackage{longtable}
\usepackage{array,graphicx}
\usepackage{graphicx,rotating,lscape}
\usepackage{threeparttable}
\usepackage{babel}
\usepackage[font=small,labelfont=bf,tableposition=top]{caption}
\usepackage{booktabs}
\usepackage{appendix}
\usepackage{makecell}
\usepackage{mathabx}
\usepackage{tabularx}
\usepackage{varwidth}
\usepackage{setspace}
\usepackage{subfigure}
\usepackage{multirow}
\usepackage{booktabs}
\usepackage{pgfplots}
\usepackage{bigstrut}
\usepackage{booktabs} 
\usepackage{tikz}
\usepackage{amsmath}
\usepackage{enumitem}
\usetikzlibrary{arrows.meta,decorations.pathreplacing,calc}

\usepackage{threeparttable}
\usepackage{array}
\usepackage[super,square,compress,sort]{natbib}
\usepackage[colorlinks,
            linkcolor=blue,
            anchorcolor=blue,
            citecolor=blue]{hyperref}

\usetikzlibrary{shapes,arrows,positioning}

\usepackage{tikz}

\newtheorem{theorem}{Theorem}
\newtheorem{proposition}{Proposition}

\newtheorem{assumption}{Assumption}

\newtheorem{definition}{Definition}
\newtheorem{example}[theorem]{Example}
\newtheorem{lemma}{Lemma}

\newenvironment{remark}[1][Remark]{\noindent\textbf{#1.} }{\ \rule{0.5em}{0.5em}}

\pgfplotsset{compat=1.18}

\begin{document}

\begin{titlepage}
\vspace{-0.8in}
\title{\Huge Difference-in-Differences Under Network Interference}

\author{Zhiguo Xiao  \and Kuan Sun  \thanks{School of Management, Fudan University, Shanghai, 200433, China. (Address correspondence to: zhiguo\_xiao@fudan.edu.cn). Zhiguo Xiao acknowledges financial support from the National Natural Science Foundation of China (Grant Number: 72232002). }}
\date{\today}

\maketitle
\vspace{-0.3in}
\begin{center}
\textbf{Abstract}
\end{center}
\indent \hspace{0.65cm}
This paper develops doubly robust estimators for direct (DATT) and spillover (SATT) average treatment effects on the treated in network-based difference-in-differences (DID) designs. Unlike standard DID methods, the proposed approach explicitly accounts for treatment spillovers and high-dimensional network confounding from complex unit dependencies in networks. It introduces a novel identification condition where conditional parallel trends hold only after adjusting for high-dimensional network confounders. The estimators are shown to be consistent and asymptotically normal as network size increases, leveraging graph neural networks (GNNs) to handle nuisance functions. Simulation studies and an empirical application on U.S. county-level mask mandates’ impact on COVID-19 transmission confirm their finite-sample performance, addressing limitations of conventional DID that ignore network interference.


\end{titlepage}

\section{Introduction}
Difference-in-differences (DID) methods are widely used for policy evaluation with observational data. In its canonical form with covariates, DID relies on the conditional parallel trends assumption (CPTA), which posits that, in the absence of treatment, treated and comparison groups with identical covariates information would have followed similar trends in potential outcomes over time (see Roth et al. 2023\cite{roth2023s}). This assumption is typically justified under the stable unit treatment value assumption (SUTVA), which rules out interference between units. However, in many real world applications, such as in social, economic, or epidemiological contexts, units are interconnected through networks. In such settings, the CPTA becomes ambiguous, as a unit’s potential outcome may depend not only on its own treatment status but also on the treatment status of its neighbors. Moreover, to restore the credibility of the CPTA under network interference, it is necessary to condition on high-dimensional network confounders. This network-mediated interference generates two distinct sources of bias: (1) spillover bias, which stems from causal effects propagated through the network, and (2) confounding bias, which arises due to the endogenous structure of the network itself.

Two-way fixed-effects (TWFE) regressions are the most common implementation of DID methods in panel data settings. The panel data literature on peer effects in networks remains sparse. The existing studies, such as Bramoullé (2020)\cite{bramoulle2020peer}, typically extend TWFE regressions with only a very low-dimensional set of controls: an individual’s own covariates, the number of immediate neighbors, and the neighbors’ average characteristics. This strategy is rather restricted for two reasons. First, the parsimonious control set implicitly presumes that only first-order connections matter, thereby overlooking confounding that may arise from higher-order network links. Second, reducing neighbors’ characteristics to simple averages cannot capture the complex, potentially nonlinear channels through which these attributes affect outcomes.

To accomodate the network effects, in this paper we decompose the average treatment effects on the treated (ATT) into two components: the direct average treatment effects on the treated (DATT) and the  spillover average treatment effects on the treated (SATT). We develop nonparametric estimation and inference procedures for both DATT and SATT under a new set of network-based conditional parallel trends assumption. To eliminate spillover bias, we adapt the exposure mapping framework to delineate the subsets of units whose untreated outcomes are expected to follow parallel paths.  To remove confounding bias, we impose the parallel trends assumption conditional on the entire covariate matrix \(\boldsymbol{X}\) and the full adjacency matrix \(\boldsymbol{A}\), thereby avoiding ad hoc restrictions to low-order neighborhoods.  We further demonstrate, both analytically and in simulations, that conventional DID estimators which ignore either treatment spillovers or network confounder can suffer substantial bias and lead to invalid inference.

Our primary contribution is to provide a theoretical foundation for difference-in-differences estimators that accommodate both treatment and confounder interference in observational networks. We extend the approximate-neighborhood-interference (ANI) framework of Leung (2022\cite{leung2022causal}, 2024\cite{leung2024graph}) to panel and repeated cross-sectional data with staggered treatment adoption.  Moreover, we enrich the emerging double/debiased-machine-learning DID literature by replacing parametric first-step models with graph-neural-network (GNN) learners that exploit the full adjacency matrix. We also demonstrate that the doubly robust DATT and SATT estimator exhibits asymptotic normality as the network size increases. Notably, this network data structure yields a distinct variance-covariance matrix. For variance estimation, we employ the network heteroskedasticity and autocorrelation consistent (HAC) estimator developed by Kojevnikov et al. (2021)\cite{kojevnikov2021limit}.

Most DID methodology continues to impose SUTVA, thus excluding spillovers, for example, the augmented IPW estimator of Sant’Anna and Zhao (2020)\cite{sant2020doubly} and the multi-period heterogeneity frameworks of de Chaisemartin and D’Haultfoeuille (2020)\cite{de2020two}, Sun and Abraham (2021)\cite{sun2021estimating}, and Callaway and Sant’Anna (2021)\cite{callaway2021difference}.  A small but growing strand relaxes SUTVA by introducing limited interference.  Butts (2021)\cite{butts2021did2s} and Fiorini (2024)\cite{fiorini2024simple} modify two-way fixed-effects (TWFE) specifications to allow local spatial spillovers; Hettinger et al. (2023)\cite{hettinger2024doubly} and Lee et al. (2023)\cite{lee2025policy} use specific exposure mappings to motivate outcome regression (OR), inverse probability weighting (IPW), and doubly robust (DR) estimators; Shahn et al. (2022) derive structural-nested mean models under clustered/network interference;  and Xu (2023)\cite{xu2023difference} adopts a design-based approach with ANI, focusing solely on outcome interference and ignoring neighbor covariate effects.  We contribute to this literature by developing a DID framework that simultaneously accommodates network interference arising from both treatment assignment and confounding variables, yielding a more comprehensive and flexible structure for causal inference with panel and repeated cross-sectional data under interference.

The remainder of the paper is organized as follows. Section 2 introduces the modeling framework, provides motivation, and defines the causal estimands of interest. Section 3 presents the main identification assumptions and examines the bias of the naive DID estimator in the presence of treatment and confounder interference, motivating the construction of a doubly robust estimand. Section 4 outlines the estimation procedure, including the use of graph neural networks (GNNs) for first-step nuisance function estimation. Section 5 establishes the large sample properties of the proposed estimators, including consistency and asymptotic normality under the ANI framework, and introduces a HAC variance estimator adapted to network dependence. Section 6 reports results from a comprehensive simulation study and Section 7 applies the method to evaluate the impact of U.S. county-level mask-mandate policy on COVID-19 transmission. Section 8 concludes.

\section{Problem Setup}
Let the population of units be \( \mathcal{N}_n = \{1,\dots,n\} \). We represent the undirected network by an \(n \times n\) binary adjacency matrix \(\boldsymbol{A}\).  A link between units \(i\) and \(j\) is indicated by \(A_{ij}=A_{ji}=1\), while self-ties are excluded by setting \(A_{ii}=0\).  The graph distance \(\ell_{\mathbf{A}}(i,j)\) is the length of the shortest path connecting nodes \(i\) and \(j\) (taken as \(\infty\) if no path exists).  For each node \(i\), its \(K\)-neighborhood is \(\mathcal{N}(i,K)=\{\,j:\ell_{\mathbf{A}}(i,j)\le K\,\}\) whose size is \(n(i,K)=|\mathcal{N}(i,K)|\). We call the nodes in \(\mathcal{N}(i,1)\setminus\{i\}\) the neighbors of \(i\) and those in \(\mathcal{N}(i,K)\setminus\{i\}\) with \(K>1\) its higher-order neighbors; the degree of node \(i\) is \(n(i,1)\), the number of its direct neighbors.

Units are indexed by \( i \in  \mathcal{N}_n \) and time periods are indexed by \( t = \{1, \dots, T \}\). $Y_{it}$ denotes the observed outcome. $D_{it}$ denotes the treatment, with its realized value $d_{it}\in \{0,1\}$. $X_{i}$ is a vector of pre-treatment covariates — such as age, geographic location, or socioeconomic status — which may influence both treatment assignment and potential outcomes. The potential outcome is  $Y_{i t}(\boldsymbol{d}_t)$, where $\boldsymbol{d}_t=\left(d_{i t}, \boldsymbol{d}_{-i, t}\right)$, with \( \boldsymbol{d}_{-i,t} \) being the treatment assignments of all other units at time \( t \). Thus the vector \( \boldsymbol{d}_t = (d_{it}, \boldsymbol{d}_{-i,t}) \) represents the full treatment assignment at time \( t \). We assume that the potential outcome is determined by 
\begin{equation}\label{1}
   Y_{i t}(\boldsymbol{d_t}) = h_{it}\left( d_{i t}, \boldsymbol{d}_{-i, t}, \boldsymbol{X}, \boldsymbol{A}, \boldsymbol{\varepsilon}_t\right),
\end{equation}
where $h_{it}$ is an unknown function, $\boldsymbol{X}=(X_{1},\dots, X_{n})'$, $\boldsymbol{\varepsilon}_{t}=(\varepsilon_{1t}, \dots, \varepsilon_{nt})'$, with $\varepsilon_{it}$'s being unobservable random errors related to the variation of the potential outcomes. We also assume the following treatment assignment mechanism:
\begin{equation}\label{2}
    D_{it} = l_{it}\bigl( \boldsymbol{X}, \boldsymbol{A}, \boldsymbol{\nu}_t\bigr),
\end{equation}
where $l_{it}$ is an unknown function, $\boldsymbol{\nu}_{t}=(\nu_{1t}, \dots, \nu_{nt})'$, with $\nu_{it}$'s being unobservable random errors related to the variation of the treatment assignment.

This setup captures potential spillovers and local interactions: an individual’s outcome may depend not only on their own treatment but also on the treatments and characteristics of neighbors in the network. In a standard DID setup, researchers often treat the treatment assignment as given or quasi-exogenous. However, when treatment is suspected to be endogenous or correlated with underlying characteristics, it can be useful to explicitly model the treatment assignment function like (\ref{2}). In this extended DID settings, incorporating a propensity score model offers two main advantages. First, when dealing with high-dimensional covariates or complex network structures, balancing treatment and control groups based solely on the outcome model becomes challenging. A propensity score model allows researchers to flexibly model the treatment assignment mechanism, using machine learning tools such as random forests, neural networks, or graph neural networks, thereby improving the accuracy of causal effect estimation. Second, within a doubly robust DID framework, the inclusion of a propensity score model provides robustness: consistent estimation and valid inference can still be achieved even if either the outcome model or the treatment model is misspecified. 

For ease of exposition, we focus on the two-period scenario, i.e., $t=1,2$, in the following analysis. The results for multi-period settings are discussed in the Appendix.

\subsection{Motivation}
Under SUTVA, the parallel trends assumption serves as the core identification condition in standard DID analysis. It states that, in the absence of treatment, the average outcome paths of the treatment and control groups would have followed the same trend over time. Formally, for untreated potential outcomes \( Y_{i2}(0) \), this implies:
\begin{equation}\label{4}
  \mathbb{E}[Y_{i2}(0) - Y_{i1}(0) \mid D_i=1] 
= 
\mathbb{E}[Y_{i2}(0) - Y_{i1}(0) \mid D_i=0].  
\end{equation}

A stronger and more flexible version is the conditional parallel trends assumption, which allows for systematic differences in observed covariates \( X_i \). It posits that, conditional on \( X_i \), the potential outcome paths of the treated and control units would have remained parallel in the absence of treatment. That is, for all relevant values of \( x \),
\begin{equation}\label{5}
   \mathbb{E}\left[Y_{i2}(0) - Y_{i1}(0) \mid D_i = 1, X_i = x\right] = 
\mathbb{E}\left[Y_{i2}(0) - Y_{i1}(0) \mid D_i = 0, X_i = x\right]. 
\end{equation}

To relax the SUTVA assumption and allow for network interference, the existing literature introduces the concept of effective treatment or exposure mapping, where each unit’s outcomes are depend not only on their own treatment status but also on the treatment received by others in their network. As formalized by Manski (2013)\cite{manski2013identification} and Aronow and Samii (2017)\cite{aronow2017estimating}, this approach defines a low-dimensional exposure vector:
\begin{equation}
   T_{i} = (D_{i},G_i)=\left(D_{i},\ g(i,\boldsymbol{D}_{-i},\boldsymbol{A})\right), 
\end{equation}
where \( g(\cdot) \) summarizes the expose to peer treatment based on the network structure \( \boldsymbol{A} \). The individual treatment \( D_i \) is separated from the exposure term to distinguish the direct treatment effect and spillover effects in the potential outcomes framework. In parallel, covariate exposure is captured through a low-dimensional control vector:

\begin{equation}
W_{i} = q(i,\boldsymbol{X},\boldsymbol{A}),
\end{equation}
which aggregates relevant covariate information from \( i \)’s neighborhood. A commonly used example of such mappings is

\begin{equation}\label{8}
 T_i=\left(D_i, \sum_{j=1}^n A_{ij} D_j\right), \quad 
W_i=\left(X_i,  \frac{\sum_{j=1}^n A_{ij} X_j}{\sum_{j=1}^n A_{ij}}\right),   
\end{equation}
where the second element of \( T_i \) captures the total number of treated neighbors, and the second element of \( W_i \) represents the average covariate value among them. Motivated by the use of low-dimensional exposure mappings in conventional cross-sectional studies to address network interference, we can immediately extend the parallel trends assumption to settings with network interference. Specifically, we assume that, conditional on network-adjusted covariates \( W_i \), the evolution of untreated potential outcomes is comparable across units with and without exposure to treatment. Formally, the assumption is stated as:
\begin{equation}\label{9}
  \mathbb{E}\left[Y_{i2}(0, {0}) - Y_{i1}(0, {0}) \mid D_i = 1,G_i = g, W_i\right]
=
\mathbb{E}\left[Y_{i2}(0, {0}) - Y_{i1}(0, {0}) \mid D_i = 0, G_i= 0, W_i\right].  
\end{equation}

As specified in (\ref{8}), the treatment vector \(\boldsymbol{D}\) reduces to two sufficient statistics: an indicator  $D_i$ for the unit’s own treatment and the count of its treated neighbors $G_i$—the former pinpoints the direct effect, while the latter captures spillovers. Likewise, \(W_i\) is summarized by the unit’s covariates and those of its immediate neighbors. Consistent with most exposure mappings literature, this construction depends only on \(\boldsymbol{D}_{\mathcal{N}(i, 1)}\) and on \(\boldsymbol{X}_{\mathcal{N}(i, 1)}\), thereby ruling out interference beyond the first-order neighborhood. Essentially, the assumption (\ref{9}) states that, conditional on a unit’s own covariates and those of its immediate neighbors, the untreated potential outcome trend of treated and control units with no treated neighbors would have evolved in parallel. However, the assumption that the summary statistics \(T_i\) and \(W_i\) can be correctly specified is difficult to justify (Sävje 2024\cite{savje2024causal}). In contrast, our model (\ref{1}) and (\ref{2}) is considerably less restrictive — we do not require the correct specification of a low-dimensional function \(T_i\) of (\(\boldsymbol{D}, \boldsymbol{A}\)) to capture treatment interference, nor do we require the correct specification of a low-dimensional function \(W_i\) based on(\(\boldsymbol{D}, \boldsymbol{A}\)) to summarize confounder interference.

\subsection{Causal estimands of interest}
We consider conditional ATT estimands that are indexed by exposure mappings following the DID literature. Let $
T_i=(D_i,G_i)=(D_i, g(i,\boldsymbol{D_{-i}},\boldsymbol{A}))$, where the function \( g(\cdot) \) takes values in a finite set \( \mathcal{G} \) of possible exposure levels. For each sample size \(n\), let \(\mathcal{M}_n\subseteq\mathcal{N}_n\) denote a selected subset of units and its size is denoted by $m_n$, i.e., \(m_n = |\mathcal{M}_n|\). 

In terms of the individual treatment, we first establish its causal estimand given a specific level of the neighborhood treatment. The definition of direct average treatment effect on the treated (DATT) is:

\begin{equation}\label{10}
  \tau^{DATT}(g)
=
\frac{1}{m_n} \sum_{i \in \mathcal{M}_n}\mathbb{E}
\left[
Y_{i2}(1,g)
-
Y_{i2}(0,g)
\mid
D_i=1, G_i=g, \boldsymbol{X}, \boldsymbol{A}
\right],
\end{equation}
for $g \in \mathcal{G}$. This denotes the direct average treatment effect on the treated when the neighborhood treatment is set to level g while adjusting for high-dimensional network confounders. We restrict the comparison to a subpopulation $\mathcal{M}_n$ in order to ensure overlap assumption, as further discussed below.

Next, define the overall DATT, denoted by $\tau^{DATT}$, as the average treatment effect on the treated aggregated over the distribution of the neighborhood treatment among treated units, which is
\begin{equation}\label{11}
     \tau^{DATT}
 =
 \sum_{g \in \mathcal{G}}
 \tau^{DATT}(g)
P\left(G_i=g \mid D_i=1, \boldsymbol{X}, \boldsymbol{A}\right).
\end{equation}

We now define the spillover effects for treated units, i.e., the SATT. Specifically, we consider the SATT of having the neighborhood treatment set to level $g$ versus 0, when the individual treatment is $d$, is defined as
\begin{equation}\label{12}
 \tau^{SATT}(g ; d)
 =
 \frac{1}{m_n} \sum_{i \in \mathcal{M}_n}\mathbb{E}
 \left[
 Y_{i2}(d, g)
 -
 Y_{i2}(d, 0)
 \mid
 D_i=1, G_i=g, \boldsymbol{X}, \boldsymbol{A}
\right].
\end{equation}

The overall SATT effect when the individual treatment equals $d$ is then given by
\begin{equation}\label{13}
     \tau^{SATT}(d)
 =
 \sum_{g \in \mathcal{G}}
 \tau^{SATT}(g ; d)
 P(G_i=g \mid D_i=1, \boldsymbol{X}, \boldsymbol{A}).
\end{equation}

The direct effects $\tau^{DATT}(g)$ in (\ref{10}) and spillover effects $\tau^{SATT}(g;d)$ in (\ref{12}) compare potential outcomes for treated units under fixed values of individual and neighborhood treatment. In contrast, the overall DATT in (\ref{11}) and SATT in (\ref{13}) average these treatment effects over the distribution of the neighborhood treatment among treated units. Unlike previous studies that consider averages over hypothetical interventions (e.g., Bernoulli assignments or general stochastic interventions), our ATT estimands fix the treatment status of the treated unit and average over the observed neighborhood treatment distribution. This allows us to identify the total ATT for units who are treated and are also exposed to other units' treatment:
\begin{equation}\label{15}
  \tau^{ATT}=
 \frac{1}{m_n} \sum_{i \in \mathcal{M}_n}\sum_{g \in \mathcal{G}}\mathbb{E}
\left[
Y_{i2}(1,g)
-
Y_{i2}(0,0)
\mid
D_i=1, G_i=g, \boldsymbol{X}, \boldsymbol{A}
\right]P(G_i=g \mid D_i=1, \boldsymbol{X}, \boldsymbol{A}),
\end{equation}
Then, it is straightforward to show that this is equal to the sum of the overall DATT and SATT eﬀects: 
\begin{align}
\tau^{ATT}
& =
\frac{1}{m_n}
\sum_{i \in \mathcal{M}_n}\sum_{g \in \mathcal{G}}
\mathbb{E}
\left[
Y_{i2}(1,g)
-
Y_{i2}(0,g)
\mid
D_i=1, G_i=g, \boldsymbol{X}, \boldsymbol{A}
\right]
P(G_i=g \mid D_i=1, \boldsymbol{X}, \boldsymbol{A}) \notag\\
& +
\frac{1}{m_n}
\sum_{i \in \mathcal{M}_n}\sum_{g \in \mathcal{G}}
\mathbb{E}
\left[
Y_{i2}(0,g)
-
Y_{i2}(0,0)
\mid
D_i=1, G_i=g, \boldsymbol{X}, \boldsymbol{A}
\right]
P(G_i=g \mid D_i=1, \boldsymbol{X}, \boldsymbol{A}) \notag\\
& = \tau^{DATT} + \tau^{SATT}(0).
\end{align}

This formula shows that the overall ATT for treated units under interference consists of two parts: the direct treatment effect (DATT) capturing how their own treatment changes outcomes, and the spillover effect (SATT) reflecting how exposure to treated neighbors affects them. 

In the main body of this paper, we develop a general framework for identifying the DATT (i.e., $\tau^{DATT}(g)$). The identification of the SATT follows a parallel logic and is discussed in the Appendix.

\section{Identiﬁcation of DATT}
First, we outline a set of commonly used assumptions for identifying our key causal estimand DATT.
\begin{assumption}[\textbf{Locality of Exposure Mapping}]\label{assumption1}
There exists a fixed neighborhood size \( K \) such that a unit's exposure mapping depends only on the treatment assignments and network structure within its \( K \)-neighborhood. Specifically, for any treatment vectors \( \boldsymbol{d}, \boldsymbol{d}' \) and network structures \( \boldsymbol{A}, \boldsymbol{A}' \), we have:
\begin{equation}
    G(i, \boldsymbol{d_{-i}}, \boldsymbol{A}) = G(i, \boldsymbol{d}'_{-i}, \boldsymbol{A}')
    \quad \text{if} \quad
    \begin{cases}
        \mathcal{N}_{\boldsymbol{A}}(i, K) = \mathcal{N}_{\boldsymbol{A}'}(i, K), \\
        \boldsymbol{A}^{\mathcal{N}_{\boldsymbol{A}}(i, K)} = \boldsymbol{A}'^{\mathcal{N}{\boldsymbol{A}'}(i, K)}, \\
        \boldsymbol{d}_{-i}^{\mathcal{N}_{\boldsymbol{A}}(i, K)} = \boldsymbol{d}_{-i}'^{\mathcal{N}_{\boldsymbol{A}'}(i, K)}.
    \end{cases}
\end{equation}
\end{assumption}

This assumption ensures that exposure mapping is determined by the local network structure and treatment assignments within the \( K \)-neighborhood. This restriction is modest and consistent with the assumptions underlying most exposure mappings in prior literature.
\begin{example} The following exposure mapping satisfies Assumption \ref{assumption1}:

\[
G_i = \mathbf{1}\left\{ \sum_{j=1}^n A_{ij} D_j > 0 \right\},
\]
where \( G_i \) indicates whether unit \( i \) has at least one treated neighbor, based on the adjacency matrix \( \boldsymbol{A} \) and the treatment vector \( \boldsymbol{D} \).

\end{example}

This mapping allows us to define DATT effect, \(\tau^{DATT}(1)\), comparing treated and untreated units with treated neighbors, and  SATT effect, \( \tau^{SATT}(1;1) \) and \( \tau^{SATT}(1;0) \) measure how having at least one treated neighbor affects outcomes for treated and untreated individuals, respectively, holding own treatment status fixed.
\begin{example}
A more general exposure mapping that satisfies Assumption \ref{assumption1} is:

\[
G_i = \left(\ \sum_{j=1}^n A_{ij} D_j \right).
\]
\end{example}
This form represents one of the most commonly used exposure mappings, leveraging local treatment aggregation to facilitate the analysis of peer effects in networked settings. Additional examples of exposure mappings under network interference can be found in the literature, including Aronow and Samii (2017)\cite{aronow2017estimating}, Sävje et al. (2021)\cite{savje2021average}, and Eckles et al. (2017)\cite{eckles2017design}.
\begin{assumption}[\textbf{No Anticipation}]\label{assumption2}
Treatment occurs only in period 2, and all units remain untreated and unaffected by any spillover effects prior to this point. 
\begin{equation}
  Y_{i1}\left(\boldsymbol{d}_{i,2}, \boldsymbol{d}_{-i,2}\right) = Y_{i1}(0, \underline{0}).  
\end{equation}
\end{assumption}
This assumption implies that the potential outcomes in the pre-treatment period is the same as it would be in the absence of both treatment and spillovers. It extends the standard no-anticipation assumption by additionally ruling out any spillover effects in the pre-treatment period, under the premise that no units are treated at that time. 

We now introduce the core assumption for identifying the DATT: 
\begin{assumption}[\textbf{Network Conditional Parallel Trends}]\label{assumption3}
For each unit  \(i \in \mathcal{M}_n\),
\begin{align}
& \mathbb{E}\left(Y_{i2}\left(0, g\right) \mid D_i= 1, G_i=g, \boldsymbol{X}, \boldsymbol{A}\right)-\mathbb{E}\left(Y_{i1}(0, {0}) \mid D_i= 1, G_i=g, \boldsymbol{X}, \boldsymbol{A}\right) \notag\\
= & \mathbb{E}\left(Y_{i2}\left(0, g\right) \mid D_i= 0, G_i=g, \boldsymbol{X}, \boldsymbol{A}\right)-\mathbb{E}\left(Y_{i1}(0, {0}) \mid D_i= 0, G_i=g, \boldsymbol{X}, \boldsymbol{A}\right).
\end{align}
\end{assumption}

Although the Network Conditional Parallel Trends (NCPT) assumption shares conceptual roots with the standard conditional parallel trends assumption (\ref{5}), our framework introduces two critical innovations. First, beyond conditioning on individual covariates \(x_i\), we incorporate the full covariate matrix \(\boldsymbol{X}\) and network structure \(\boldsymbol{A}\). This generalization enables the use of network-derived covariate functions—such as centrality measures or positional characteristics—rather than relying solely on individual-level attributes. Second, while traditional parallel trends assumptions compare potential outcome trends across treatment groups absent treatment, our NCPT assumption explicitly addresses interference. By controlling for spillover exposure through the exposure mapping, we isolate the direct effect under the assumption that potential outcomes evolve similarly across exposure groups when spillover effects are accounted for.

In essence, our assumption simultaneously accommodates both treatment interference and confounding interference. To highlight the practical implications of this distinction, we subsequently demonstrate how the naive difference-in-differences estimator becomes biased—whether targeting the conventional average treatment effect on the treated (ATT) or our proposed direct average treatment effect on the treated (DATT)—when interference in treatment assignment and confounder structure is neglected.

\subsection{The bias of the naive DID estimator under treatment interference}
In this part, we examine the case where only treatment interference is present, excluding the influence of confounding interference. In the following subsection, we extend the analysis to incorporate confounding interference. Under SUTVA, the potential outcomes for ATT depend only on the individual’s own treatment status, denoted as \(Y_{i2}(d_i)\), and are unaffected by the treatment assignments of other units. The standard ATT under the SUTVA assumption is given by:  
\begin{equation}
    \tau^{SUTVA} = \frac{1}{m_n} \sum_{i \in \mathcal{M}_n} \mathbb{E}\left[Y_{i2}(1) - Y_{i2}(0) \mid D_i = 1,x_i\right]
\end{equation}
In the naive DID framework, several covariate-adjusted estimators for the ATT have been proposed, such as the outcome regression estimator, the inverse probability weighting estimator, and the doubly robust DID estimator. All these estimators consistently estimates the following quantity:

\begin{equation}
 \tau^{obs} = \frac{1}{m_n} \sum_{i \in \mathcal{M}_n} \left\{ \mathbb{E}[Y_{i2} - Y_{i1} \mid D_i = 1, x_i] - \mathbb{E}[Y_{i2} - Y_{i1} \mid D_i = 0, x_i] \right\}.  
\end{equation}

If the conditional parallel trends assumption (\ref{5}) holds, then \(\tau^{\text{obs}}\) and \({\tau^{SUTVA}}\) are identical. However, if SUTVA is violated, we cannot obtain a clean \(\tau^{\text{obs}}\) because the second-period outcome \( Y_{i2} \) would be influenced by other individuals' treatment statuses and, thus, these estimators would clearly not estimate the quantity \(\tau^{SUTVA}\). Moreover, they also fail to consistently estimate the direct average treatment
effect on the treated \(\tau^{DATT}(g)\) or \(\tau^{DATT}\), since they compare changes over time between treated and control units based solely on their own treatment status \(D_i\), while disregarding potential variation in exposure due to the neighborhood exposure.

We next present a proposition that characterizes the discrepancy between \(\tau^{obs}\) and \(\tau^{DATT}\), and identifies two primary sources of bias contributing to the difference.

\begin{proposition}\label{prop1}
Suppose Assumption \ref{assumption1}, \ref{assumption2} and \ref{assumption3} holds for any \( g \in \mathcal{G}, \forall i \). An unbiased estimator targeting  $\tau^{obs}$ does not imply unbiasedness for
$\tau^{DATT}$, the resulting bias equals
\begin{align}
\tau^{\text{obs}} - \tau^{\text{DATT}} 
= & \frac{1}{m_n} \sum_{i \in \mathcal{M}_n} \sum_{g \in \mathcal{G}} \Big[ 
\mathbb{E}\big(Y_{i2} - Y_{i1} \mid D_i = 0, G_i = g, x_i\big) - 
\mathbb{E}\big(Y_{i2} - Y_{i1} \mid D_i = 0, G_i = g', x_i\big) \Big] \notag\\
&\quad \cdot \Big[ \mathbb{P}(G_i = g \mid D_i = 0, x_i) - \mathbb{P}(G_i = g \mid D_i = 1, x_i) \Big].
\end{align}
\end{proposition}
Proposition \ref{prop1} characterizes the bias that arises when interference is mistakenly ignored. This result parallels the discussion in Forastiere et al. (2021)\cite{forastiere2021identification}, which also considers interference over networks but under an unconfoundedness framework. However, in our setting, the bias of the ATT-type estimator  vanishes under weaker conditions than those required in prior work. Specifically, it is sufficient for the neighborhood treatment \(G_i\) to have no effect on outcome changes among control units only, or for the individual and neighborhood treatments \((D_i, G_i)\) to be conditionally independent given covariates \(X_i\).

There are several main sources of dependence between \(D_i\) and \(G_i\), including: (i) unobserved neighborhood-level confounders not captured by \(X_i\), and (ii) peer influence in treatment uptake (see Forastiere et al., 2021)\cite{forastiere2021identification}. We now examine the bias of the naive DID estimator in the presence of confounder interference.

\subsection{The bias of the naive DID under confounder interference}
We are concerned with bias that arises when the parallel trends assumption fails to hold conditional on individual covariates \(X_i\), but becomes valid when conditioning additionally on a vector of neighborhood-level covariates \(U_i \in \mathcal{U}\). A typical example of \(U_i\) is the network-weighted average of neighbors' covariates, such as $U_i = \frac{\sum_{j=1}^n A_{ij} X_j}{\sum_{j=1}^n A_{ij}}$.

In what follows, we assume a network parallel trends assumption holds conditional on the enriched covariate set \((X_i, U_i)\), where \(U_i\) captures aggregated information from unit \(i\)'s neighbors:
\begin{equation}\label{22}
\mathbb{E}[Y_{i2}(0,g) - Y_{i1}(0,g) \mid D_i = 1, X_i, U_i] = \mathbb{E}[Y_{i2}(0,g) - Y_{i1}(0,g) \mid D_i = 0, X_i, U_i], \quad \forall g \in \mathcal{G}.    
\end{equation}

We present a proposition that characterizes the discrepancy between \(\tau^{\mathrm{obs}}\) and \(\tau^{\mathrm{DATT}}\) under confounder interference.
\begin{proposition}\label{prop2}
Suppose Assumption(\ref{assumption1}),(\ref{assumption2}) and (\ref{22}) holds for any \( g \in \mathcal{G}, \forall i \). An unbiased estimator targeting  $\tau^{obs}$ does not imply unbiasedness for
$\tau^{DATT}$, the resulting bias equals
\begin{align}
\tau^{\text{obs}} - \tau^{\text{DATT}} = 
& \frac{1}{m_n} \sum_{i \in \mathcal{M}_n} \sum_{g \in \mathcal{G}} \sum_{u \in \mathcal{U}}\Big[ 
\mathbb{E}\big(Y_{i2} - Y_{i1} \mid D_i = 0, G_i = g,U_i = u, x_i\big)\notag \\
&- 
\mathbb{E}\big(Y_{i2} - Y_{i1} \mid D_i = 0, G_i = g',U_i = u', x_i\big) \Big]\notag \\
&\quad \cdot \Big[ \mathbb{P}(U_i = u \mid D_i =1,G_i = g, x_i)\cdot \mathbb{P}(G_i = g \mid D_i =1, x_i)\notag \\
&\quad\quad- \mathbb{P}(U_i = u \mid D_i =0,G_i = g, x_i)\cdot \mathbb{P}(G_i = g \mid D_i =0, x_i) \Big].
\end{align}
If we further assume that \(D_i\) and \(G_i\) are conditionally independent given \(X_i\), then the bias simplifies to:
\begin{align}
\tau^{\text{obs}} - \tau^{\text{DATT}} = 
& \frac{1}{m_n} \sum_{i \in \mathcal{M}_n} \sum_{u \in \mathcal{U}} \Big[
\mathbb{E}\big(Y_{i2} - Y_{i1} \mid D_i = 0, U_i = u, x_i\big) \notag\\
&\quad - \mathbb{E}\big(Y_{i2} - Y_{i1} \mid D_i = 0, U_i = u', x_i\big) \notag\Big] \\
&\quad \cdot \Big[ \mathbb{P}(U_i = u \mid D_i =1, x_i) - \mathbb{P}(U_i = u \mid D_i =0, x_i) \Big].
\end{align}
\end{proposition}

Proposition \ref{prop2} implies that if we mistakenly assume no interference and condition only on an individual’s own covariates, the resulting bias is a combination of two sources: treatment interference bias and confounder interference bias. In contrast, even if we assume that the individual treatment \(D_i\) is independent of the neighborhood exposure \(G_i\) given a subset of covariates \(X_i\)—thus effectively ruling out treatment interference—bias may still arise due to unmeasured neighborhood-level confounders.

To address this, conditioning on a simplified summary measure of first-order neighbors' covariates—as described in (\ref{22}), such as their average covariate values—can help alleviate this confounder interference bias. Nevertheless, this method still fails to capture important structural heterogeneity in the network relationships.

Our network conditional parallel trends assumption as (\ref{assumption3}) requires conditioning on the entire adjacency matrix \(\boldsymbol{A}\) and the full covariate matrix \(\boldsymbol{X}\). This implies a much stricter version of parallel trends, as we assume that only units with isomorphic network positions and identical covariates exhibit parallel counterfactual trajectories. As shown in Figure \ref{figue2}, units 3 and 4 share identical individual-level confounders as well as the same first-order neighborhood confounder information. This configuration satisfies an analogue of the parallel trends assumption as (\ref{22}). However, our method does not require the parallel trends assumption to hold specifically between units 3 and 4. In fact, in this example, units 3 and 4 are not isomorphic (they would have been if unit 2 and unit 3 were not connected). Instead, our method requires the parallel trends assumption to hold among units with greater similarity in confounding information, enabling more accurate estimation of causal effects.
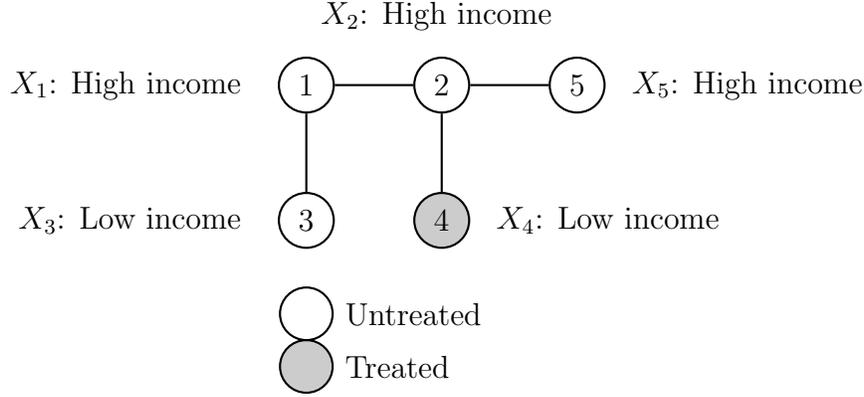
\begin{figure}[h!]
\centering
\begin{tikzpicture}[node distance=1.8cm, auto, thick]

\tikzstyle{treated}=[circle,draw,fill=black!20,minimum size=7mm]
\tikzstyle{control}=[circle,draw,fill=white!20,minimum size=7mm]

% 节点定义
\node[control] (1) {1};
\node[control] (2) [right of=1] {2};
\node[control] (3) [below of=1] {3};
\node[treated] (4) [below of=2] {4};
\node[control] (5) [right of=2] {5};

% 网络结构（邻接矩阵A）
\draw[-] (1)--(2);
\draw[-] (1)--(3);
\draw[-] (2)--(4);
\draw[-] (2)--(5);

% 协变量标注 (X)
\node[above=0.2cm of 2, align=left] { $X_2$: High income };
\node[left=0.2cm of 1, align=right] { $X_1$: High income };
\node[left=0.2cm of 3, align=right] { $X_3$: Low income };
\node[right=0.2cm of 4, align=left] { $X_4$: Low income };
\node[right=0.2cm of 5, align=right] { $X_5$: High income };
% 图例
\node[treated, below=1.2cm of 3, label=right:Treated] {};
\node[control, below=0.5cm of 3, label=right:Untreated] {};

\end{tikzpicture}

\caption{\textbf{Conditional Parallel Trends on X and A}}
 \begin{minipage}{\textwidth}
    \small \textit{Note:} Figure \ref{figue2} illustrates the logic underlying the network conditional parallel trends assumption, highlighting the necessity of conditioning on both individual covariates \(\boldsymbol{X}\) and the network adjacency structure \(\boldsymbol{A}\). 
 Initially, if we condition only on traditional individual covariates \(X_i\), units 3 and 4 would appear to satisfy parallel trends because they share similar characteristics (e.g., both having low income).

 Extending the conditioning set to include simple first-order neighborhood information—such as the average covariates of neighbors, defined by \(W_i = \left(X_i, \frac{\sum_{j=1}^n A_{ij} X_j}{\sum_{j=1}^n A_{ij}}\right)\)—may still suggest that units 3 and 4 are comparable, since their aggregated neighbor profiles appear similar.
However, when we fully condition on the entire adjacency matrix \(\boldsymbol{A}\) along with the covariates \(\boldsymbol{X}\), it becomes clear that units 3 and 4 do not satisfy the network conditional parallel trends assumption. This discrepancy arises because deeper network features—such as global connectivity patterns, centrality, or indirect pathways—are now captured. Therefore, failing to fully account for the detailed network structure encoded in \(\boldsymbol{A}\) can result in biased comparisons.

  \end{minipage}

  \label{figue2}
\end{figure}

\subsection{The doubly robust estimand}
Under our network parallel trends assumption, the causal estimand of interest for both DATT and SATT can be transformed into an identifiable estimand using one of three commonly used strategies in the literature: outcome regression, inverse probability weighting, and doubly robust methods. These approaches are widely discussed in works such as Abadie (2005)\cite{abadie2005semiparametric}, Wooldridge (2009)\cite{woolridge2009introductory}, and Sant’Anna and Zhao (2020)\cite{sant2020doubly}. Among them, doubly robust methods are particularly appealing due to their robustness to model misspecification. Moreover, doubly robust frameworks are naturally compatible with modern machine learning techniques, allowing researchers to flexibly model high-dimensional network covariates while still maintaining valid inference.

To formally construct the doubly robust estimand, we first define the outcome regression function as

\begin{equation}\label{26}
  \mu_{t, dg}\left(i, \boldsymbol{X}, \boldsymbol{A}\right)=\mathbb{E}\left(Y_{i t} \mid D_i=d, G_i=g,  \boldsymbol{X}, \boldsymbol{A}\right). \quad  
\end{equation}
We also follow Imbens (2000)\cite{imbens2000role} to define the generalized propensity score regression as 
\begin{equation}\label{27}
   p_{dg}(i, \boldsymbol{X}, \boldsymbol{A}) = \mathbb{P}(D_i = d,\, G_i = g \mid \boldsymbol{X}, \boldsymbol{A}).
\end{equation}
Here we focus on the panel data case with $t=1,2$, leaving the derivations for the multiple-period panel and the repeated cross-section cases to the Appendix. Let \(\Delta Y_i = Y_{i2} - Y_{i1}\) and \(\Delta\mu_{dg}\left(i, \boldsymbol{X}, \boldsymbol{A}\right)= \mu_{2,dg}\left(i, \boldsymbol{X}, \boldsymbol{A}\right) - \mu_{1,dg}\left(i, \boldsymbol{X}, \boldsymbol{A}\right)\). Define
\begin{equation}\label{28}
   \tau^{dr}(g) = \frac{1}{m_n} \sum_{i \in \mathcal{M}_n} \tau^{dr}_i(g),
\end{equation}
where
\begin{equation}
  \tau^{dr}_i(g) = \left( D_i \mathbf{1}\{G_i = g\} - \frac{(1 - D_i)\mathbf{1}\{G_i = g\} \cdot p_{1g}(i, \boldsymbol{X}, \boldsymbol{A})}{1 - p_{1g}(i, \boldsymbol{X}, \boldsymbol{A})} \right) \left( \Delta Y_i - \Delta \mu_{0g}(i, \boldsymbol{X}, \boldsymbol{A}) \right).
\end{equation}

\begin{remark}
The estimand \( \tau^{dr}(g) \) defined above bears a close resemblance to the doubly robust score for panel data DID models developed by Sant’Anna and Zhao (2020)\cite{sant2020doubly}. Their results demonstrate that estimators based on this doubly robust structure achieve semiparametric efficiency under standard regularity conditions. Our work generalizes this framework by (i) incorporating spillover effects and (ii) accounting for network confounding – features absent in their original formulation.

Beyond efficiency considerations in conventional panel settings (the focus of Sant’Anna and Zhao), we investigate whether valid inference persists when combining the doubly robust score with double/debiased machine learning (DML) techniques. This approach parallels that of Chang (2020)\cite{chang2020double}, who examines DML-based inference in standard DID frameworks. Our contribution adapts this methodology to settings with dependence structures induced by network interference, while establishing a central limit theorem that ensures asymptotically valid inference even with flexibly estimated, high-dimensional nuisance parameters.
\end{remark}

\begin{proposition}\label{prop3}
Suppose Assumptions (\ref{assumption1})-(\ref{assumption3}) hold. If either the conditional outcome mean model or the propensity score model is correctly specified, then $\tau^{DATT}(g)=\tau^{dr}(g)$.
\end{proposition}
Proposition \ref{prop3} establishes that, under mild assumptions on the exposure mapping and conditional trends, our primary target--the causal estimand \(\tau^{DATT}(g)\)--is identified (i.e., expressible in terms of observable quantities) provided that either the outcome regression model or the propensity score model is correctly specified. Consequently, our proposed estimand $\tau^{dr}(g)$ is doubly robust, yielding valid inference even under misspecification of one of the two models. Relative to approaches relying solely on outcome regression or inverse probability weighting, the doubly robust estimand imposes weaker assumptions and demonstrates greater reliability in practice.

\section{Estimation}
\subsection{Network DR-DID estimator}
To estimate the doubly robust estimand \(\tau^{dr}(g)\), we adopt a standard plug-in approach that leverages machine learning-based estimators for the relevant nuisance components. Specifically, let \(\hat{p}_{dg}(i, \boldsymbol{X}, \boldsymbol{A})\) denote the estimated generalized propensity score, \(\hat{\mu}_{t,dg}(i, \boldsymbol{X}, \boldsymbol{A})\) denote the estimated outcome regression for \(t=1,2\), and \(\Delta\hat{\mu}_{dg}\left(i, \boldsymbol{X}, \boldsymbol{A}\right)= \hat{\mu}_{2,dg}\left(i, \boldsymbol{X}, \boldsymbol{A}\right) - \hat{\mu}_{1,dg}\left(i, \boldsymbol{X}, \boldsymbol{A}\right)\). We then propose the following DR-DID estimator for \(\tau^{dr}(g)\) that allows for network interference:

\[
\hat{\tau}^{dr}(g) = \frac{1}{m_n} \sum_{i \in \mathcal{M}_n} \hat{\tau}^{dr}_i(g),
\]
where each \(\hat{\tau}^{dr}_i(g)\) is defined as:
\begin{align}\label{30}
\hat{\tau}^{dr}_i(g) = & \left(\left({D_i\mathbf{1}\{G_i=g\}}  \right) -\frac{(1-D_i)\mathbf{1}\{G_i=g\} \hat{p}_{1g}(i, \boldsymbol{X}, \boldsymbol{A})}{1-\hat{p}_{1g}(i, \boldsymbol{X}, \boldsymbol{A})} \right)\left(\Delta Y_{i} - \Delta \hat{\mu}_{0g}(i, \boldsymbol{X}, \boldsymbol{A})\right).
\end{align}

The validity of the Network DR-DID estimator $\hat{\tau}^{dr}(g)$ hinges on accurately estimating the nuisance components, the propensity score and the outcome regression function. These components are traditionally estimated using parametric models, such as logistic regression or linear outcome regression. However, such models often lack the flexibility needed to capture the complex dependencies and nonlinear interactions that arise in networked data, especially when spillovers effects or network confounding are present. To address this, we propose using Graph Neural Networks (GNNs) to estimate these nuisance functions. GNNs are a class of nonparametric machine learning models designed specifically for graph-structured data. They incorporate both individual-level covariates (node features) and network structure (neighborhood edges) to generate learned representations of each unit. Through iterative local averaging—commonly referred to as message passing—GNNs extract information from a unit's neighbors to improve prediction. This allows GNNs to flexibly approximate high-dimensional, nonlinear relationships in both the outcome and treatment assignment mechanisms, without requiring explicit model specification or manual feature construction. As a result, plugging GNNs estimator into the doubly robust framework can improve both the robustness and efficiency of causal inference in complex network settings.

In the next subsection, we provide a brief overview of the GNN architecture used to estimate the nuisance components.
\subsection{GNNs estimator for nuisance functions}
GNNs are deep learning models designed to model graph-structured data. A standard GNNs architecture consists of nested, parameterized, vector-valued functions called neurons arranged in \( L \) layers. The embedding of the \( i \)-th node at layer \( l \), denoted \( h_i^{(l)} \), is updated via the following message-passing architecture for layers \( l=1, \ldots, L \):
\begin{equation}
h_i^{(l)} = \Phi_{0l}\left(h_i^{(l-1)}, \Phi_{1l}\left(h_i^{(l-1)}, \{h_j^{(l-1)} : j \in \mathcal{N}(i)\}\right)\right),
\end{equation}
where \( \Phi_{0l}(\cdot) \) and \( \Phi_{1l}(\cdot) \) are parameterized, vector-valued functions. The embedding is initialized as \( h_i^{(0)} = x_i \), thus initially incorporating only node features. As layers progress, the embeddings incorporate increasingly richer neighborhood information. This architecture endows GNNs with several essential properties. \textit{Permutation invariance} ensures that the aggregation of neighbor embeddings is insensitive to the order in which they appear, a critical feature given that graph neighborhoods are inherently unordered. Due to the unordered nature of graph neighborhoods, the estimation functions \( p_{1g}(i, \boldsymbol{X}, \boldsymbol{A}) \) and \( \Delta\mu_{0g}(i, \boldsymbol{X}, \boldsymbol{A}) \) must be permutation invariant in the features of \( i \)'s neighbors. This ensures that the estimated values are not sensitive to arbitrary ordering of the neighborhood set. As discussed in Leung (2024)\cite{leung2024graph}, such invariance allows us to reduce from a collection of neighbor-specific functions to a single symmetric function over the neighborhood multiset.
\textit{Neighborhood scope} is controlled via the number of layers $L$, such that the final node embedding $h_i^{(L)}$ reflects information from the node’s $L$-hop neighborhood. The \textit{scalability} of GNNs arises from the fact that the learnable functions $\Phi_{0l}(\cdot)$ and $\Phi_{1l}(\cdot)$ depend solely on the dimension of node features and are independent of the graph size. As a result, GNNs can be deployed efficiently across networks of varying scales, from small graphs to large-scale systems. These structural properties make GNNs particularly effective for accurately modeling complex dependencies in observational data, thereby improving the estimation of nuisance functions such as propensity scores or conditional outcome regressions. By incorporating rich relational structure in a robust and scalable manner, GNNs estimator can more effectively adjust for confounding, particularly when outcomes or treatments exhibit network-dependent relationship. The specific choices of $\Phi_{0l}(\cdot)$ and $\Phi_{1l}(\cdot)$ define the architectural variants of the GNN and thus influence both model expressiveness and computational behavior. There exist various GNN embedding architectures, including the Graph Convolutional Network (GCN) \cite{kipf2017semi}, the Graph Isomorphism Network (GIN) \cite{xu2019powerful}, and the Principal Neighborhood Aggregation (PNA) network \cite{corso2020principal}, each differing in how neighborhood information is aggregated and combined.

We define \(\mathcal{F}_{\mathrm{GNN}}^{\mathrm{prop}}(L)\) and \(\mathcal{F}_{\mathrm{GNN}}^{\mu}(L)\) as classes of \(L\)-layer graph neural networks used to estimate the generalized propensity score and the outcome regression function, respectively. For any \(f \in \mathcal{F}_{\mathrm{GNN}}^{\mathrm{prop}}(L)\) or \(f \in \mathcal{F}_{\mathrm{GNN}}^{\mu}(L)\), we let \(f(i, \boldsymbol{X}, \boldsymbol{A})\) denote its output for node \(i\), corresponding to the final-layer embedding \(h_i^{(L)}\). The nuisance estimators \(\hat{f}_{\mathrm{GNN}}^{\mathrm{prop}}\) and \(\hat{f}_{\mathrm{GNN}}^{\mu}\) are obtained via empirical risk minimization:

\[
\hat{f}_{\mathrm{GNN}}^{\mathrm{prop}} \in \arg\min_{f \in \mathcal{F}_{\mathrm{GNN}}^{\mathrm{prop}}(L)} \sum_{i=1}^n \ell_{\mathrm{log}}\left(\mathbf{1}\{D_i = d,\, G_i = g\},\ f(i, \boldsymbol{X}, \boldsymbol{A})\right),
\]

\[
\hat{f}_{\mathrm{GNN}}^{\mu} \in \arg\min_{f \in \mathcal{F}_{\mathrm{GNN}}^{\mu}(L)} \sum_{i: D_i = 0,\, G_i = g} \ell_{\mathrm{sq}}\left(\Delta Y_{i},\ f(i, \boldsymbol{X}, \boldsymbol{A})\right),
\]
where \(\ell_{\mathrm{log}}(y, \hat{y}) = -y\hat{y} + \log(1 + e^{\hat{y}})\) is the logistic loss and \(\ell_{\mathrm{sq}}(y, \hat{y}) = 0.5(y - \hat{y})^2\) is the squared-error loss.

The estimated functions are then used to define the nuisance components:
\[
\hat{p}_{dg}(i, \boldsymbol{X}, \boldsymbol{A}) = \frac{\exp\left(\hat{f}_{\mathrm{GNN}}^{\mathrm{prop}}(i, \boldsymbol{X}, \boldsymbol{A})\right)}{1 + \exp\left(\hat{f}_{\mathrm{GNN}}^{\mathrm{prop}}(i, \boldsymbol{X}, \boldsymbol{A})\right)},
\]
\[
\hat{\mu}_{0g}(i, \boldsymbol{X}, \boldsymbol{A}) = \hat{f}_{\mathrm{GNN}}^{\mu}(i, \boldsymbol{X}, \boldsymbol{A}).
\]

The estimated nuisance functions \(\hat{p}_{dg}(i, \boldsymbol{X}, \boldsymbol{A})\) and \(\hat{\mu}_{0g}(i, \boldsymbol{X}, \boldsymbol{A})\) are then plugged into the doubly robust score defined in (\ref{30}) to deliver \(\hat{\tau}^{dr}_{i}(g), i=1,\dots, m_{n}\), which are averaged to obtain the estimator \(\hat{\tau}^{dr}(g)\) for the estmation of \(\tau^{dr}(g)\), and consequently the DATT \(\tau^{DATT}(g)\).

\section{Asymptotic theory}
\subsection{Limiting distribution}
We first discuss the limiting distribution of \(\hat{\tau}^{dr}(g)\) as \( n \to \infty \). While our analysis treats \(\bigl(\boldsymbol{X}, \boldsymbol{A}, \boldsymbol{\varepsilon}_t,\boldsymbol{\nu}_t)\) as random, the asymptotic theory conditions on \(\bigl(\boldsymbol{X}, \boldsymbol{A}\bigr)\) to avoid imposing additional assumptions on their underlying dependence structure. Define:

\begin{align}
\phi_i(g) = & \left( D_i \mathbf{1}\{G_i = g\} - \frac{(1 - D_i)\mathbf{1}\{G_i = g\} \, p_{1g}(i, \boldsymbol{X}, \boldsymbol{A})}{1 - p_{1g}(i, \boldsymbol{X}, \boldsymbol{A})} \right) \left( \Delta Y_i - \Delta u_{0g}(i, \boldsymbol{X}, \boldsymbol{A}) \right)-\tau^{DATT}(g),
\end{align}
and
\begin{equation}
  \sigma_n^2=\operatorname{Var}\left(\left.\frac{1}{\sqrt{m_n}} \sum_{i \in \mathcal{M}_n} \phi_i(g) \right\rvert\, \boldsymbol{X}, \boldsymbol{A}\right).  
\end{equation}
The following assumptions are required to guarantee the validity of our asymptotic analysis. They ensure that the estimators are well-defined and converge properly as the sample size increases.

\begin{assumption}[\textbf{Approximate Neighborhood Interference}]\label{assumption4}
For each sample size \(n\in\mathbb{N}\) there exist non-negative functions \(\gamma_{n}(s)\) and \(\eta_{n}(s)\), defined on \(\mathbb{R}_{+}\), satisfying  
\[
\sup_{n\in\mathbb{N}}\max\{\gamma_{n}(s),\eta_{n}(s)\}\xrightarrow[s\to\infty]{}0,
\]
such that for every individual \(i\in\mathcal{N}_{n}\) and period \(t\),
\[
\begin{aligned}
\mathbb{E}\!\Bigl[\,\bigl|h_{it}(\mathbf{D}_{t},\mathbf{X},\mathbf{A},\boldsymbol{\varepsilon}_{t})
        -h_{it}\bigl(\mathbf{D}^{\mathcal{N}(i,s)}_{t},\mathbf{X}^{\mathcal{N}(i,s)},
        \mathbf{A}^{\mathcal{N}(i,s)},\boldsymbol{\varepsilon}^{\mathcal{N}(i,s)}_{t}\bigr)\bigr|
        \,\bigm|\,\mathbf{D}_{t},\mathbf{X},\mathbf{A}\Bigr]
        &\le\gamma_{n}(s), \\[6pt]
\mathbb{E}\!\Bigl[\,\bigl|l_{it}(\mathbf{X},\mathbf{A},\boldsymbol{\nu}_{t})
        -l_{it}\bigl(\mathbf{X}^{\mathcal{N}(i,s)},
        \mathbf{A}^{\mathcal{N}(i,s)},\boldsymbol{\nu}^{\mathcal{N}(i,s)}_{t}\bigr)\bigr|
        \,\bigm|\,\mathbf{X},\mathbf{A}\Bigr]
        &\le\eta_{n}(s).
\end{aligned}
\]
\end{assumption}

Assumption \ref{assumption4} indicates a uniform, distance-based decay of interference within the network. Specifically, outcome models \(h_{it}\) and the propensity score model \(l_{it}\) are asymptotically insensitive to information originating beyond an \(s\)-step neighborhood of the focal node. The bounding functions \(\gamma_{n}(s)\) and \(\eta_{n}(s)\) converge to zero uniformly in \(n\), ensuring that remote nodes exert a vanishing influence as \(s\to\infty\). Hence, observations from a single, expansive network can be treated as only weakly dependent, permitting the application of classical asymptotic theorem. Leung (2022)\cite{leung2022causal} demonstrated that the ANI assumption is satisfied by a range of interference structures, such as the linear-in-means model with endogenous peer effects.

\begin{assumption}[\textbf{Moments and Overlap}]\label{assumption5}
(a) There exist constants \(M<\infty\) and \(p>4\) such that, for every sample size \(n\in\mathbb N\), every individual \(i\in\mathcal M_n\), every period t, and every treatment vector \(\boldsymbol d_t\in\{0,1\}^n\),
\[
\mathbb E\!\bigl[\,|Y_{i2}(\boldsymbol d_t)|^{\,p}\,\bigm|\,\mathbf X,\mathbf A\bigr]\;\le\;M
\quad\text{a.s.}
\]

(b) For every unit \( i \in \mathcal{M}_n \), each treatment status \( d \in \{0,1\} \), and each exposure level \( g \in \mathcal{G} \), there exists a constant \( \varepsilon > 0 \) such that
\[
\varepsilon < p_{dg}(i, \boldsymbol{X}, \boldsymbol{A}) < 1 - \varepsilon.
\]

\end{assumption}

Assumption \ref{assumption5}(a) bound \(p\)-th moments of the potential outcomes, which is a standard regularity condition, see the double machine learning literature (e.g., Chernozhukov et al., 2018\cite{chernozhukov2018double}; Farrell, 2018\cite{farrell2018deep}; Farrell et al., 2021\cite{farrell2021deep}). By contrast, Assumption \ref{assumption5}(b) is conceptually more restrictive because it links the exposure mapping, network structure, and treatment distribution. For simplicity, I assume that the overlap condition holds for every unit in the population. However, under certain specifications of the exposure mapping, this assumption might not always be valid. If violations appear, one can (i) redefine or coarsen the exposure mapping, (ii) trim units with near-zero or near-one propensities, or (iii) restrict inference to a subpopulation where credible overlap holds.

% 若尚未加载 enumitem，请在导言区加入
% \usepackage{enumitem}

\begin{assumption}[\textbf{GNN Convergence Rates}]\label{assumption6}
For every $g\!\in\!\mathcal{G}$ and $d\!\in\!\{0,1\}$, let
$\hat p_{dg}(i,\mathbf X,\mathbf A)$ and
$\Delta\hat\mu_{dg}(i,\mathbf X,\mathbf A)$
be the first–stage GNN estimators of
$p_{dg}(i,\mathbf X,\mathbf A)$ and
$\Delta\mu_{dg}(i,\mathbf X,\mathbf A)$, respectively. Suppose the following conditions hold:
\begin{enumerate}[label=(\alph*)]

\item
\[
\begin{aligned}
\frac{1}{m_n}\sum_{i\in\mathcal M_n}
\bigl(\hat p_{dg}(i,\mathbf X,\mathbf A)-p_{dg}(i,\mathbf X,\mathbf A)\bigr)^2
&=o_p(1),\\[4pt]
\frac{1}{m_n}\sum_{i\in\mathcal M_n}
\bigl(\Delta\hat\mu_{dg}(i,\mathbf X,\mathbf A)-\Delta\mu_{dg}(i,\mathbf X,\mathbf A)\bigr)^2
&=o_p(1).
\end{aligned}
\]

\item
\[
\biggl\{
      \frac{1}{m_n}\sum_{i\in\mathcal M_n}
      (\hat p_{dg}-p_{dg})^2
\biggr\}^{1/2}
\,
\biggl\{
      \frac{1}{m_n}\sum_{i\in\mathcal M_n}
      (\Delta\hat\mu_{dg}-\Delta\mu_{dg})^2
\biggr\}^{1/2}
= o_p\!\bigl(n^{-1/2}\bigr).
\]

\item
\[
\frac{1}{m_n}\sum_{i\in\mathcal M_n}
\biggl\{\frac{\hat p_{dg}(i,\mathbf X,\mathbf A)-\mathbf{1}(D_i=d)\mathbf{1}(G_i=g)}{1-\hat p_{dg}(i,\mathbf X,\mathbf A)}
(\Delta u_{dg}(i,\mathbf X,\mathbf A)-\Delta\hat u_{dg}(i,\mathbf X,\mathbf A))\biggr\}
= o_p\!\bigl(n^{-1/2}\bigr).
\]
\end{enumerate}
\end{assumption}

These regularity conditions are well-established in the double machine learning literature. The validity of these assumptions are verified for convolutional neural networks (CNNs) in the i.i.d. setting by both Farrell (2021)\cite{farrell2021deep} and Ghasempour et al. (2024)\cite{ghasempour2024convolutional}. Extending to network data, Wang et al. (2024)\cite{wang2024graph} establish analogous \( n^{-1/2} \)-rate convergence results for GNNs, under certain architectural constraints. A recent advance by Leung (2024)\cite{leung2024graph} strengthens the required independence structure through an approximate conditional-independence assumption (his Assumption 8).  Under that assumption he shows that, for both the propensity-score model and the outcome-regression model,

\[
\frac{1}{m_n}\sum_{i\in\mathcal M_n}
\Bigl[f_t(i,\mathbf X,\mathbf A)-
      f_t\!\bigl(i,\mathbf X^{\mathcal N(i,L)},
                    \mathbf A^{\mathcal N(i,L)}\bigr)\Bigr]^2
        =o_p\!\bigl(n^{-1/2}\bigr),
\]
where the neighborhood-restricted target \(f_t\!\bigl(i,\mathbf X^{\mathcal N(i,L)},\mathbf A^{\mathcal N(i,L)}\bigr)\)  
can be consistently estimated by an \(L\)-layer GNN.  Leung further shows that choosing  
\(L\asymp \log n\) (or any slowly diverging sequence) is sufficient for the network to achieve the required approximation error while keeping the effective model complexity low. Together, these results place GNNs on essentially the same theoretical footing as classical machine-learning estimators for semiparametric causal inference.

To establish a central limit theorem for our main term, we require that the sequence \(\{\phi_{i}(g)\}_{i=1}^n\) be \(\psi\)-dependent (as in Definition C.1 by Kojevnikov et al., 2021\cite{kojevnikov2021limit}). This assumption restricts how quickly a specific dependence measure decays in relation to the growth rate of network neighborhoods. To formalize this, we first define the \(s\)-neighborhood boundary of node \(i\) as
\[
\mathcal{N}^{\partial}(i, s) = \{j \in \mathcal{N}_n: \ell(i, j) = s\},
\]
and its \(k\)th moment by
\[
\delta_n^{\partial}(s; k) = \frac{1}{n} \sum_{i=1}^n |\mathcal{N}^{\partial}(i, s)|^k.
\]
Next, we introduce
\[
\Delta_n(s, m; k) = \frac{1}{n} \sum_{i=1}^n \max_{j \in \mathcal{N}^{\partial}(i, s)} \left|\mathcal{N}(i, m) \setminus \mathcal{N}(j, s-1)\right|^k,
\]
which captures, on average, the maximal expansion of a node's \(m\)-neighborhood beyond the \((s-1)\)-neighborhood of any node on its \(s\)-boundary. Based on this, we define
\[
c_n(s, m; k) = \inf_{\alpha > 1} \Delta_n(s, m; k\alpha)^{1/\alpha} \delta_n^{\partial}\Bigl(s; \frac{\alpha}{\alpha-1}\Bigr)^{1-1/\alpha},
\]
which is a quantity that essentially measures the network density. Finally, we set
\begin{equation}\label{34}
   \psi_n(s) = \max_{i \in \mathcal{N}_n}\Bigl(\gamma_n(s/2)+\eta_n(s/2)\,[1+n(i,K)+\Lambda_n(i,s/2)\,n(i,s/2)]\Bigr), 
\end{equation}
where \(\Lambda_n(i, s/2)\) is a constant defined in the subsequent assumption. \(\psi_n(s)\) provides a bound on the covariance between \(\phi_{g}(i)\) and \(\phi_{g}(j)\) when the network distance \(\ell_{\boldsymbol{A}}(i, j)\) is at most \(s\).

\begin{assumption}[\textbf{Weak Dependence for CLT}] \label{assumption7} 
(a) The dependence coefficients are uniformly bounded. Specifically, \(\sup_{n\in\mathbb{N}}\max_{s\ge1}\psi_n(s)<\infty\) almost surely.  
(b)
Let \(p > 4\) in assumption \ref{assumption5}(b), for some sequence \(v_n \to \infty\) and for each \(k \in \{1,2\}\), the following conditions hold:
\[
\frac{1}{n^{k/2}} \sum_{s=0}^\infty c_n(s, v_n; k)\,\psi_n(s)^{1-(2+k)/p} \to 0,\quad n^{3/2}\,\psi_n(v_n)^{1-1/p} \to 0,
\]
and
\[
\limsup_{n \to \infty} \sum_{s=0}^\infty \delta_n^{\partial}(s; 2)^{1/2}\,\gamma_n(s/2)^{1-2/p} < \infty \quad \text{a.s.}
\]

\end{assumption}

The quantity \( \psi_n(s) \) measures the degree of dependence between pairs of observations \( Y_{it} \) and \( Y_{jt} \) across different individuals \( i \ne j \) at the same time period \( t \).
 As discussed in Kojevnikov et al. (2021)\cite{kojevnikov2021limit}, many network-dependent processes satisfy $\psi$-dependence. Moreover, Leung (2024, Appendix C)\cite{leung2024graph} derives the rate of \(\psi_n(s)\) under the ANI assumption (\ref{assumption4}) for observational data. This enables the use of robust inferential procedures despite the presence of approximate local network dependencies in the observational setting. The first two parts of Condition (b) in Assumption (\ref{assumption7}) coincide with Condition ND in Kojevnikov et al. (2021)\cite{kojevnikov2021limit}. The third part is an analogous requirement that guarantees the linear expansion of the doubly robust ATT estimator under network dependence. Leung (2024)\cite{leung2024graph} demonstrates that both polynomial and exponential neighborhood-growth patterns satisfy all three components of Condition (b).
\setcounter{theorem}{0}
\begin{theorem}\label{theorem1}
     Under Assumptions \ref{assumption1}-\ref{assumption7}, the Network DR-DID estimator $\hat{\tau}^{dr}\left(g\right)$ has asymptotically normal distribution centered around $\tau^{DATT}\left(g\right)$. Specifically, 

$$
\sigma_n^{-1 / 2} \sqrt{m_n}\left(\hat{\tau}^{dr}\left(g\right)-\tau^{DATT}\left(g\right)\right) \xrightarrow{d} \mathbf{N}(0,1) .
$$
\end{theorem}

\subsection{Variance Estimation}
We now focus on the variance estimator for large‐sample inference.  To estimate the asymptotic variance, we utilize the network HAC (heteroskedasticity and autocorrelation consistent) estimator as described by Kojevnikov et al. (2021)\cite{kojevnikov2021limit}:

\begin{equation}
\hat{\sigma}^2 = \frac{1}{m_n} \sum_{i \in \mathcal{M}_n} \sum_{j \in \mathcal{M}_n} \left(\hat{\tau}^{dr}_i(g) - \hat{\tau}^{dr}(g)\right)\left(\hat{\tau}^{dr}_j(g) - \hat{\tau}^{dr}(g)\right) \mathbf{1}\{\ell_{\boldsymbol{A}}(i, j) \leq B_n\}.  
\end{equation}

We adopt the uniform-kernel variance estimator and choose the bandwidth as  

\begin{equation}
B_n \;=\;
\begin{cases}
\Bigl\lceil\dfrac{1}{2+\gamma}\,\mathcal{L}(\boldsymbol A)\Bigr\rceil, 
& \text{if } \displaystyle
   \mathcal{L}(\boldsymbol A)\;<\;
   2\,\dfrac{\log n}{\log \bar\delta(\boldsymbol A)},\\[8pt]
\displaystyle
\Bigl\lceil\bigl[\mathcal{L}(\boldsymbol A)\bigr]^{\tfrac{1}{2+\gamma}}\Bigr\rceil,
& \text{otherwise},
\end{cases}
\end{equation}
where \(\lceil\cdot\rceil\) denotes rounding up to the nearest integer; \(\bar\delta(\boldsymbol A)=\dfrac{1}{n}\sum_{i,j}A_{ij}\) is the network’s average degree; \(\mathcal{L}(\boldsymbol A)\) is the average path length; and \(\gamma>0\) is a fixed positive constant.

Thus, the bandwidth adapts to the network’s size and density while accounting for first-stage estimation error. This bandwidth rule builds on the scheme proposed by Leung (2022)\cite{leung2022causal}, Leung (2024)\cite{leung2024graph}. The next theorem states the asymptotic properties of \(\hat{\sigma}^{2}\).  
Because we condition on \((\mathbf X,\mathbf A)\), \(\hat{\sigma}^{2}\) is not guaranteed to be consistent—exactly as in Leung (2022)\cite{leung2022causal}.  Nevertheless, the same argument shows it is typically asymptotically conservative.  We introduce the required notation below:

$$
\mathcal{J}_n(s, m)=\left\{(i, j, k, l) \in \mathcal{N}_n^4: k \in \mathcal{N}(i, m), l \in \mathcal{N}(j, m), \ell_{\boldsymbol{A}}(i, j)=s\right\}.
$$

\begin{assumption}[\textbf{Weak Dependence for \(\hat{\sigma}\)}]\label{assumption8}
\begin{enumerate}[label=(\alph*)]
\item  For some \(\epsilon\in(0,1)\) and a bandwidth \(B_n\to\infty\),
$
\lim_{n\to\infty}
\frac{1}{n}\sum_{s=0}^{\infty}
c_n\!\left(s,B_n;2\right)\,
\psi_n(s)^{\,1-\epsilon}
\;=\;0
\quad\text{a.s.}
$

\item  
$\frac{1}{n}\sum_{i=1}^{n} n\!\left(i,B_n\right)
\;=\;o_p\!\bigl(\sqrt{n}\bigr).
$

\item  

$
\frac{1}{n}\sum_{i=1}^{n} n\!\left(i,B_n\right)^{2}
\;=\;O_p\!\bigl(\sqrt{n}\bigr).
$

\item  

$\sum_{s=0}^{n}
\bigl|\mathcal J_n\!\left(s,B_n\right)\bigr|\,
\psi_n(s)
\;=\;o_p\!\bigl(n^{2}\bigr).$

\end{enumerate}
\end{assumption}
This assumption regulates the growth rate of the neighborhood size and the bandwidth $B_n$, ensuring that the estimator $\hat{\sigma}^2$ remains consistent and well-behaved in large samples by balancing bias and variance. Assumption \ref{assumption8}(a) corresponds to the first part of Assumption \ref{assumption7}(b). Parts (b)–(d) align with Assumptions 7(b)–(d) in Leung (2022)\cite{leung2022causal}, which serve to characterize the bias properties of the variance estimator. These conditions are satisfied under both polynomial and exponential neighborhood growth network.

\begin{theorem}
Define ${\phi}^*_{i}(g)$ by replacing $\tau^{DATT}\left(g\right)$ in the definition of $\phi_{i}(g)$ with $\tau_i^{DATT}\left(g\right)$. Let

$$
\begin{aligned}
\hat{\sigma}_*^2 & =\frac{1}{m_n} \sum_{i \in \mathcal{M}_n} \sum_{j \in \mathcal{M}_n} {\phi}^*_{i}(g) {\phi}^*_{j}(g) \mathbf{1}\left\{\ell_{\boldsymbol{A}}(i, j) \leqslant B_n\right\} \quad \text { and } \\
R_n & =\frac{1}{m_n} \sum_{i \in \mathcal{M}_n} \sum_{j \in \mathcal{M}_n}\left(\tau_i^{DATT}\left(g\right)-\tau^{DATT}\left(g\right)\right)\left(\tau_j^{DATT}\left(g\right)-\tau^{DATT}\left(g\right)\right) \mathbf{1}\left\{\ell_{\boldsymbol{A}}(i, j) \leqslant B_n\right\}.
\end{aligned}
$$
Under Assumption \ref{assumption8} and the assumptions of Theorem \ref{theorem1}, we have that
$$
\hat{\sigma}^2=\hat{\sigma}_*^2+R_n+o_p(1) \quad \text { and } \quad\left|\hat{\sigma}_*^2-\sigma_n^2\right| \xrightarrow{p} 0 .
$$
\end{theorem}
This extends Proposition 4.1 of Kojevnikov et al. (2021)\cite{kojevnikov2021limit} and Theorem 4 of Leung (2022)\cite{leung2022causal} to accommodate doubly robust ATT estimators. Note that \( R_n \) is a HAC estimator of the variance of the unit-level contrasts \( \tau_i^{DATT}(g) \), in which case \( \hat{\sigma}^2 \) would be asymptotically conservative.

\section{Simulations}
In this simulation, we demonstrate the finite-sample performance of the estimators proposed for DATT. For the data generating process, we simulated a network \(\boldsymbol{A} \) comprising 2000 individuals based on a random geometric graph model, which defines the adjacency matrix \(\boldsymbol{A} \) by setting
\[
A_{ij} = \mathbf{1}\left\{\|\rho_i - \rho_j\| \leq r_n\right\},
\]
where the positions \(\{\rho_i,\rho_j\}_{i=1}^{n}\) are independently and uniformly drawn from the unit square \([0,1]^2\), and the radius parameter \( r_n \) is specified as \( r_n = \sqrt{5/(\pi n)} \). The simulated random geometric graph has an average path length of approximately 39.4.

We consider a two-period panel data structure, with the outcome equation for the first period generated as:
\begin{equation}\label{36}
  Y_{\text{pre}, i} = 0.5
+ \frac{\sum_{j=1}^{n} A_{ij} X_j}{\sum_{j=1}^{n} A_{ij}} 
+ X_i 
+ \epsilon_i 
+ \frac{\sum_{j=1}^{n} A_{ij} \epsilon_j}{\sum_{j=1}^{n} A_{ij}},  
\end{equation}
where \(\{X_i\}_{i=1}^n\) are i.i.d. draws from a discrete uniform distribution on \(\{0, 0.25, 0.5, 0.75, 1\}\), and \(\{\epsilon_i\}_{i=1}^n\) are i.i.d. \(\mathbf{N}(0,1)\) random variables. The treatment variable \( D_i \) is generated according to the following equation:
\[
D_i = \mathbf{1}\left\{
0.5 
+ 1.5 \frac{\sum_{j=1}^{n} A_{ij} D_j}{\sum_{j=1}^{n} A_{ij}} 
+  \frac{\sum_{j=1}^{n} A_{ij} X_j}{\sum_{j=1}^{n} A_{ij}} 
- X_i 
+ \nu_i 
+ \frac{\sum_{j=1}^{n} A_{ij} \nu_j}{\sum_{j=1}^{n} A_{ij}} > 0
\right\},
\]
where the error terms \(\{\nu_i\}_{i=1}^{n}\) are i.i.d. as \(\mathbf{N}(0,1)\). The outcome equation for the post-treatment period is defined as
\[
Y_{\text{post}, i} = 0.5 
+ 0.8 \frac{\sum_{j=1}^{n} A_{ij} Y_j}{\sum_{j=1}^{n} A_{ij}} 
+ 10 \frac{\sum_{j=1}^{n} A_{ij} X_j}{\sum_{j=1}^{n} A_{ij}} 
+ X_i 
+ \mu_i 
+ \frac{\sum_{j=1}^{n} A_{ij} \mu_j}{\sum_{j=1}^{n} A_{ij}}
\]
The error terms \( \{u_i\} \) are also i.i.d. as standard normal. The true value of the estimand $\tau^{DATT}$ is zero under this design. 

We compare two estimators: GNNs and nonparametric generalized linear model (NGLM) estimators. The GNNs are implemented using the PNA architecture~\cite{corso2020principal}, with the number of layers \( L \in \{1, 2, 3\} \).  
Both \( \phi_0^{(l)} \) and \( \phi_1^{(l)} \) are single-layer multilayer perceptrons (MLPs) with hidden dimension \( H \in \{1, 3, 5\} \). As for NGLM, we apply polynomial basis expansions of degree 1, 2, or 3 to estimate the nuisance functions. The degree of the polynomial plays a role analogous to the number of layers 
L in GNNs, as both determine the order of neighborhood effects captured by the model.

\begin{table}[htbp]
\centering
\caption{Simulation Results for GNNs and NGLM}
\resizebox{\textwidth}{!}{
\begin{tabular}{cccccccccc}
\toprule
& \multicolumn{3}{c}{$L=1$} & \multicolumn{3}{c}{$L=2$} & \multicolumn{3}{c}{$L=3$} \\
\cmidrule(lr){2-4} \cmidrule(lr){5-7} \cmidrule(lr){8-10}
$n$ &  & 2000 &  &  & 2000 &  &  & 2000 &  \\
\# treated         &  & 1105 &  &  & 1105 &  &  & 1105 &  \\
$H$                & 1 & 3 & 5 & 1 & 3 & 5 & 1 & 3 & 5 \\
\midrule

$\hat{\tau}^{dr}_{GNN}$ & 0.0469 & 0.0652 & 0.0110 & 0.0366 & 0.0128 & 0.0012 & 0.0192 & 0.0172 & 0.0320 \\
SE & 0.1564 & 0.1766 & 0.1787 & 0.1556 & 0.1854 & 0.1674 & 0.1588 & 0.1511 & 0.2502 \\

CI & 0.8080 & 0.9440 & 0.9540 & 0.8140 & 0.9660 & 0.8800 & 0.8160 & 0.8360 & 0.8660 \\
SE IID& 0.0979 & 0.1079 & 0.1108 & 0.0972 & 0.1161 & 0.1136 & 0.1009 & 0.0937 & 0.2049 \\
CI IID             & 0.5968 & 0.6934 & 0.6442 & 0.6332 & 0.5675 & 0.6955 & 0.5841 & 0.5980 & 0.5321  \\ 
\midrule

$\hat{\tau}^{dr}_{NGLM}$ &  & 0.0836 &  &  & 0.0764 &  &  & 0.0774 &  \\ 
SE            &  & 0.098 &  &  & 0.098  & &  & 0.098 &  \\

\bottomrule
\end{tabular}}
\label{maintab}
\end{table}

Table \ref{maintab} presents simulation results based on 1000 replications for the random geometric graph. The upper panel reports results with nuisance parameters estimated by the GNNs method, and the lower panel reports results with nuisance parameters estimated the NGLM method employing polynomial sieve methods, where the polynomial order is also indicated by $L$. For convenience, we refer to the former as the GNN method and the latter as the NGLM method. The row labeled $\hat{\tau}^{dr}_{GNN}$ reports the average value of the $\hat{\tau}^{dr}_{GNN}$ estimates, and 
the row labeled $\hat{\tau}^{dr}_{NGLM}$ reports the average value of the $\hat{\tau}^{dr}_{NGLM}$ estimates, both of which also reflect the bias due to the fact that the true parameter $\tau^{DATT}$ is zero. ``SE'' denotes standard error constructed by the HAC estimator. ``CI'' displays the coverage rate of confidence interval constructed with the HAC variance estimator. ``SE IID" denotes standard error computed under the assumption of independence and identical distribution, with ``CI IID" being the coverage rate of confidence interval constructed with the i.i.d. variance estimator.

The bias results presented in the first row of Table \ref{maintab} demonstrate that the GNNs method provides reliable causal estimates across all specifications of $L$ and $H$. Notably, GNNs with 
$L=2$ layers consistently outperform other configurations, achieving the lowest bias regardless of the hidden dimension $H$. Furthermore, for a fixed number of layers, bias tends to decrease as the hidden dimension increases.

The HAC standard errors are substantially larger than those computed under the i.i.d. assumption, suggesting the presence of both heteroskedasticity and autocorrelation in the error terms. Coverage rates generally improve with larger hidden dimensions, and our proposed method produces more accurate confidence intervals compared to those derived from i.i.d. standard errors. However, as is common with HAC-type estimators, our confidence intervals exhibit a slight degree of undercoverage.

The NGLM method also delivers reliable causal estimates for all choices of polynomial order, though the magnitude of its bias is larger than that of the GNNs method. This suggests that GNNs capture a different function of \((\boldsymbol{X}, \boldsymbol{A})\) than the \(W_i\) variables alone, one that better adjusts for confounding effects.

\section{An Application}
We employ the method proposed in this paper to assess the impact of mask mandate policy on the spread of COVID-19 in the US. Our analysis is based on a balanced panel constructed from data used by Chernozhukov et al. (2021b)\cite{chernozhukov2021association}, which consist of 2,510 US counties observed weekly from April 1, 2020 to December 2, 2020. A total of 736 counties remained untreated throughout the study period and thus serve as the control group. To accommodate variation in treatment timing across the remaining counties and maintain a clean $2 \times 2$ DID design, we select a subset of counties for the treatment group. Specifically, we focus on the 343 counties that adopted mask mandates in week 28—the week with the highest number of implementations. We define the pre-treatment period as all weeks before week 28 ($t=1$) and the post-treatment period as week 28 onward ($t=2$).

In this study, \( Y_{it} \) denotes the logarithm of reported COVID-19 cases in county \( i \) at time period \( t \). The main treatment variable \( D_{it} \) represents a policy indicator for mask mandates. The set of control variables \( X_{i} \) includes measures of foot traffic to K-12 schools and colleges (sourced from SafeGraph), along with other policy indicators such as stay-at-home orders and bans on gatherings of more than 50 persons, as well as the weekly growth rate in COVID-19 testing. We construct an adjacency matrix based on the geographic distance matrix between counties identified by their FIPS codes, where a link is assumed to exist between two counties if the distance between them is less than 400 kilometers, and no link otherwise. 

For illustration, we consider the estimation of the direct treatment effect \( \tau^{\text{DATT}}(1) \). 
We compare two approaches: our proposed Network DR-DID and the DR-DID of Sant’Anna and Zhao (2020)\cite{sant2020doubly}. Both estimators target a direct effect of mask mandate policy, but differ in how they account for network-related interference. Our estimator conditions on having at least one treated neighbor and explicitly controls for the network confounding spillovers, while the DR-DID method assumes no treatment and confounding spillovers. 

\begin{table}[htbp]
\centering
\caption{Comparison of ATT Estimates}
\begin{tabular}{lcc}
\toprule
Method & ATT Estimate & Standard Error \\
\midrule
Network DR-DID       & -0.7021$^{***}$ & 0.2983 \\
DR-DID & -0.9363$^{***}$ & 0.2586 \\
\bottomrule
\end{tabular}
\begin{tablenotes}
\small
\item \textit{Notes:} Robust standard errors are reported in parentheses.\\
$^{***}p<0.01$, $^{**}p<0.05$, $^{*}p<0.1$
\end{tablenotes}
\label{tab:att_comparison}
\label{table2}
\end{table}

As shown in Table \ref{table2}, both estimation methods yield significantly negative causal effect values, indicating that the mask mandate policy effectively and significantly reduced the number of COVID-19 cases. It is important to note that the magnitude of the Network DR-DID estimate is smaller than that of the conventional DR-DID estimate. This finding aligns with our intuition that the protective effect of wearing a mask diminishes when one’s neighbors also wear masks, compared to scenarios where neighbors do not. This discrepancy suggests that traditional estimates may suffer from bias due to unaccounted spillover effects and network confounders.

\section{Conclusion}
In this article, we develop doubly robust estimators for the Direct Average Treatment Effect on the Treated (DATT) and the Spillover Average Treatment Effect on the Treated (SATT) in network-based DID designs, where conditional parallel trends hold after adjusting for high-dimensional network confounders. The proposed estimators remain consistent for the DATT (or SATT) under the condition that either the propensity score model or the outcome regression model is correctly specified. We establish their large-sample properties and demonstrate that, under mild regularity conditions, the doubly robust estimators are asymptotically normal as the network size increases. The practical utility of our method is illustrated through Monte Carlo simulations and an empirical application.

Our findings can be extended to several other settings of practical relevance. First, the network-based analytical framework developed in this study can be adapted to alternative identification strategies in panel data settings, particularly those relying on sequential conditional independence assumptions. Second, while our analysis focuses on contemporaneous treatment effects, incorporating both dynamic treatment effects and spillover effects simultaneously would introduce additional methodological challenges. Finally, the current framework assumes a static network structure, whereas real-world networks often exhibit dynamic evolution. Extending the causal inference framework to account for network dynamics—such as by modeling network formation or selection processes over time—represents a promising direction for future research.
\newpage
\bibliographystyle{plainnat}
\bibliography{reference}

\newpage

% Table generated by Excel2LaTeX from sheet 'First'
\section*{Appendix A: Proofs of Results}

\subsection*{Proof of  Proposition 1:}

\begin{proof}
\[
\begin{aligned} 
\tau^{\text{obs}} 
&= \frac{1}{m_n} \sum_{i \in \mathcal{M}_n}  \Big[  \mathbb{E}\big(Y_{i2} - Y_{i1} \mid D_i =1, x_i\big) - \mathbb{E}\big(Y_{i2} - Y_{i1} \mid D_i =0, x_i\big) \Big] \\ 
&= \frac{1}{m_n} \sum_{i \in \mathcal{M}_n} \sum_{g \in \mathcal{G}} \mathbb{E}\big(Y_{i2} - Y_{i1} \mid D_i =1, G_i = g, x_i\big) \cdot \mathbb{P}(G_i = g \mid D_i =1, x_i) \\
&\quad - \frac{1}{m_n} \sum_{i \in \mathcal{M}_n} \sum_{g \in \mathcal{G}} \mathbb{E}\big(Y_{i2} - Y_{i1} \mid D_i =0, G_i = g, x_i\big) \cdot \mathbb{P}(G_i = g \mid D_i =0, x_i). \\
&\quad \text{(by iterated expectations law)}
\end{aligned}
\]

\[
\begin{aligned} 
\tau^{\text{DATT}} 
&= \frac{1}{m_n} \sum_{i \in \mathcal{M}_n} \sum_{g \in \mathcal{G}} \mathbb{E}\big(Y_{i2}(1,g) - Y_{i2}(0,g) \mid D_i =1, G_i = g, x_i\big) \cdot \mathbb{P}(G_i = g \mid D_i =1, x_i) \\ 
&= \frac{1}{m_n} \sum_{i \in \mathcal{M}_n} \sum_{g \in \mathcal{G}} \Big[ 
\mathbb{E}\big( Y_{i2}(1, g) - Y_{i1}(0, 0) \mid D_i = 1, G_i = g, x_i \big) \\
&\quad - \mathbb{E}\big( Y_{i2}(0, g) - Y_{i1}(0, 0) \mid D_i = 0, G_i = g, x_i \big) \Big] \cdot \mathbb{P}(G_i = g \mid D_i = 1, x_i) \\
&\quad \text{(by network conditional parallel trends)} \\
&= \frac{1}{m_n} \sum_{i \in \mathcal{M}_n} \sum_{g \in \mathcal{G}} \mathbb{E}\big(Y_{i2} - Y_{i1} \mid D_i =1, G_i = g, x_i\big) \cdot \mathbb{P}(G_i = g \mid D_i =1, x_i) \\
&\quad - \frac{1}{m_n} \sum_{i \in \mathcal{M}_n} \sum_{g \in \mathcal{G}} \mathbb{E}\big(Y_{i2} - Y_{i1} \mid D_i =0, G_i = g, x_i\big) \cdot \mathbb{P}(G_i = g \mid D_i =1, x_i). \\
&\quad \text{(by consistency and the no anticipation assumption)}
\end{aligned}
\]
Then, we have
\[
\begin{aligned}
\tau^{\text{obs}} - \tau^{\text{DATT}} 
= & \frac{1}{m_n} \sum_{i \in \mathcal{M}_n} \sum_{g \in \mathcal{G}} 
\mathbb{E}\big(Y_{i2} - Y_{i1} \mid D_i = 0, G_i = g, x_i\big) \\
&\quad \cdot \Big[ \mathbb{P}(G_i = g \mid D_i = 1, x_i) - \mathbb{P}(G_i = g \mid D_i = 0, x_i) \Big] \\
= & \frac{1}{m_n} \sum_{i \in \mathcal{M}_n} \sum_{g \in \mathcal{G}} \Big[ 
\mathbb{E}\big(Y_{i2} - Y_{i1} \mid D_i = 0, G_i = g, x_i\big) - 
\mathbb{E}\big(Y_{i2} - Y_{i1} \mid D_i = 0, G_i = g', x_i\big) \Big] \\
&\quad \cdot \Big[ \mathbb{P}(G_i = g \mid D_i = 1, x_i) - \mathbb{P}(G_i = g \mid D_i = 0, x_i) \Big].
\end{aligned}
\]
\[
\begin{aligned}
&\quad \text{(since the subtraction of a constant baseline term }\mathbb{E}(Y_{i2} - Y_{i1} \mid D_i = 0, G_i = g', x_i) \\
&\qquad  \text{ leaves the expression unchanged}) \\
\end{aligned}
\]
\end{proof}

\subsection*{Proof of  Proposition 2:}

\begin{proof}
\[
\begin{aligned} 
\tau^{\text{obs}} 
&= \frac{1}{m_n} \sum_{i \in \mathcal{M}_n}  \Big[  \mathbb{E}\big(Y_{i2} - Y_{i1} \mid D_i =1, x_i\big) - \mathbb{E}\big(Y_{i2} - Y_{i1} \mid D_i =0, x_i\big) \Big] \\ 
&= \frac{1}{m_n} \sum_{i \in \mathcal{M}_n} \sum_{g \in \mathcal{G}}\sum_{u \in \mathcal{U}} \mathbb{E}\big(Y_{i2} - Y_{i1} \mid D_i =1, G_i = g, U_i = u,x_i\big) \\
    &\quad\quad\quad\cdot \mathbb{P}(U_i = u \mid D_i =1,G_i = g, x_i)\cdot \mathbb{P}(G_i = g \mid D_i =1, x_i) \\
&\quad - \frac{1}{m_n} \sum_{i \in \mathcal{M}_n} \sum_{g \in \mathcal{G}} \sum_{u \in \mathcal{U}}\mathbb{E}\big(Y_{i2} - Y_{i1} \mid D_i =0, G_i = g,U_i = u, x_i\big)  \\
    &\quad\quad\quad\cdot \mathbb{P}(U_i = u \mid D_i =0,G_i = g, x_i)\cdot \mathbb{P}(G_i = g \mid D_i =0, x_i). \\
&\quad \text{(by iterated expectations law)}
\end{aligned}
\]

\[
\begin{aligned} 
\tau^{\text{DATT}} 
&= \frac{1}{m_n} \sum_{i \in \mathcal{M}_n} \sum_{g \in \mathcal{G}} \sum_{u \in \mathcal{U}} \mathbb{E}\big(Y_{i2}(1,g) - Y_{i2}(0,g) \mid D_i =1, G_i = g,U_i = u, x_i\big) \\
&\quad\quad\quad\cdot \mathbb{P}(U_i = u \mid D_i =1,G_i = g, x_i)\cdot \mathbb{P}(G_i = g \mid D_i =1, x_i) \\ 
&= \frac{1}{m_n} \sum_{i \in \mathcal{M}_n} \sum_{g \in \mathcal{G}}\sum_{u \in \mathcal{U}} \Big[ 
\mathbb{E}\big( Y_{i2}(1, g) - Y_{i1}(0, 0) \mid D_i = 1, G_i = g,U_i = u, x_i \big) \\
&\quad - \mathbb{E}\big( Y_{i2}(0, g) - Y_{i1}(0, 0) \mid D_i = 0, G_i = g, U_i = u,x_i \big) \Big] \\
&\quad\quad\quad\cdot \mathbb{P}(U_i = u \mid D_i =1,G_i = g, x_i)\cdot \mathbb{P}(G_i = g \mid D_i =1, x_i) \\
&\quad \text{(by network conditional parallel trends)} \\
&= \frac{1}{m_n} \sum_{i \in \mathcal{M}_n} \sum_{g \in \mathcal{G}}\sum_{u \in \mathcal{U}} \mathbb{E}\big(Y_{i2} - Y_{i1} \mid D_i =1, G_i = g,U_i = u, x_i\big) \\
&\quad\quad\quad\cdot \mathbb{P}(U_i = u \mid D_i =1,G_i = g, x_i)\cdot \mathbb{P}(G_i = g \mid D_i =1, x_i) \\
&\quad - \frac{1}{m_n} \sum_{i \in \mathcal{M}_n} \sum_{g \in \mathcal{G}}\sum_{u \in \mathcal{U}} \mathbb{E}\big(Y_{i2} - Y_{i1} \mid D_i =0, G_i = g,U_i = u, x_i\big) \\
&\quad\quad\quad\cdot \mathbb{P}(U_i = u \mid D_i =1,G_i = g, x_i)\cdot \mathbb{P}(G_i = g \mid D_i =1, x_i). \\
&\quad \text{(by consistency and the no anticipation assumption)}
\end{aligned}
\]
Then, we have
\[
\begin{aligned}
\tau^{\text{obs}} - \tau^{\text{DATT}} 
= & \frac{1}{m_n} \sum_{i \in \mathcal{M}_n} \sum_{g \in \mathcal{G}} \sum_{u \in \mathcal{U}}
\mathbb{E}\big(Y_{i2} - Y_{i1} \mid D_i = 0, G_i = g,U_i=u, x_i\big) \\
&\quad \cdot \Big[ \mathbb{P}(U_i = u \mid D_i =1,G_i = g, x_i)\cdot \mathbb{P}(G_i = g \mid D_i =1, x_i) \\
&\quad\quad- \mathbb{P}(U_i = u \mid D_i =0,G_i = g, x_i)\cdot \mathbb{P}(G_i = g \mid D_i =0, x_i) \Big] \\
= & \frac{1}{m_n} \sum_{i \in \mathcal{M}_n} \sum_{g \in \mathcal{G}}\sum_{u \in \mathcal{U}} \Big[ 
\mathbb{E}\big(Y_{i2} - Y_{i1} \mid D_i = 0, G_i = g,U_i = u, x_i\big) \\
&- 
\mathbb{E}\big(Y_{i2} - Y_{i1} \mid D_i = 0, G_i = g',U_i = u', x_i\big) \Big] \\
&\quad \cdot \Big[ \mathbb{P}(U_i = u \mid D_i =1,G_i = g, x_i)\cdot \mathbb{P}(G_i = g \mid D_i =1, x_i) \\
&\quad\quad- \mathbb{P}(U_i = u \mid D_i =0,G_i = g, x_i)\cdot \mathbb{P}(G_i = g \mid D_i =0, x_i) \Big]. \\
\end{aligned}
\]
\[
\begin{aligned}
&\quad \text{(since the subtraction of a constant baseline term }\mathbb{E}(Y_{i2} - Y_{i1} \mid D_i = 0, G_i = g',U_i = u', x_i) \\
&\qquad  \text{ which do not depend on $g$ and $u$.}) \\
\end{aligned}
\]
Then, under conditional independence between \(Z_i\) and \(G_i\), the bias become:
\[
\begin{aligned}
\tau^{\text{obs}} - \tau^{\text{DATT}} 
= & \frac{1}{m_n} \sum_{i \in \mathcal{M}_n} \sum_{g \in \mathcal{G}} \sum_{u \in \mathcal{U}}
\mathbb{E}\big(Y_{i2} - Y_{i1} \mid D_i = 0, G_i = g,U_i=u, x_i\big) \\
&\quad \cdot \Big[ \mathbb{P}(U_i = u \mid D_i =1,G_i = g, x_i)\cdot \mathbb{P}(G_i = g \mid  x_i) \\
&\quad\quad- \mathbb{P}(U_i = u \mid D_i =0,G_i = g, x_i)\cdot \mathbb{P}(G_i = g \mid  x_i) \Big].
\end{aligned}
\]
After marginalizing over \(G_i\), the expression simplifies to
\[
\begin{aligned}
\tau^{\text{obs}} - \tau^{\text{DATT}} = 
& \frac{1}{m_n} \sum_{i \in \mathcal{M}_n} \sum_{u \in \mathcal{U}} \Big[
\mathbb{E}\big(Y_{i2} - Y_{i1} \mid D_i = 0, U_i = u, x_i\big) \\
&\quad - \mathbb{E}\big(Y_{i2} - Y_{i1} \mid D_i = 0, U_i = u', x_i\big) \Big] \\
&\quad \cdot \Big[ \mathbb{P}(U_i = u \mid D_i =1, x_i) - \mathbb{P}(U_i = u \mid D_i =0, x_i) \Big].
\end{aligned}
\]
\end{proof}

\subsection*{Proof of  Proposition 3:}
\begin{proof}
Recall that:
   \[
\begin{aligned}
{\tau}(g) = \mathbb{E}_D\left(\left({D_i\mathbf{1}\{G_i=g\}}  \right) -\frac{(1-D_i)\mathbf{1}\{G_i=g\} {p}_{1g}(i, \boldsymbol{X}, \boldsymbol{A})}{1-{p}_{1g}(i, \boldsymbol{X}, \boldsymbol{A})} \right)\left(\Delta Y_{i} - \Delta {u}_{0g}(i, \boldsymbol{X}, \boldsymbol{A})\right).
\end{aligned}
\]
For notational simplicity, \( \mathbb{E}_D \) denote the finite population expectation conditional on $\boldsymbol{X}$ and $\boldsymbol{A}.$

\textbf{Case 1: When outcome regression models are correctly speciﬁed.}
In this case, we have that $ \Delta {\mu}_{0t}(i, \boldsymbol{X}, \boldsymbol{A})= \Delta m_{0t}(i, \boldsymbol{X}, \boldsymbol{A})$ a.s., i.e. the outcome regression models are correctly specified.

\[
\begin{aligned}
{\tau}^{dr}(g)
&= \mathbb{E}_{D} \left[ \left( D_i \mathbf{1}\{G_i=g\} - \frac{(1 - D_i)\mathbf{1}\{G_i=g\}{p}_{1g}(i,\boldsymbol{X},\boldsymbol{A})}{1 - {p}_{1g}(i,\boldsymbol{X},\boldsymbol{A})} \right)\left( \Delta Y_i - \Delta {m}_{0g}(i,\boldsymbol{X},\boldsymbol{A}) \right) \right] \\[10pt]
&= \mathbb{E}_{D}\left[ D_i \mathbf{1}\{G_i=g\}\left(\Delta Y_i - \Delta {m}_{0g}(i,\boldsymbol{X},\boldsymbol{A})\right)\right] \\[5pt]
&\quad-\;\mathbb{E}_{D}\left[\frac{(1 - D_i)\mathbf{1}\{G_i=g\}{p}_{1g}(i,\boldsymbol{X},\boldsymbol{A})}{1 - {p}_{1g}(i,\boldsymbol{X},\boldsymbol{A})}\left(\Delta Y_i - \Delta {m}_{0g}(i,\boldsymbol{X},\boldsymbol{A})\right)\right] \\[12pt]
&= \mathbb{E}_{D}\left[D_i \mathbf{1}\{G=g\}\left(\Delta Y - \Delta {m}_{0g}(i,\boldsymbol{X},\boldsymbol{A})\right)\right] \\[5pt]
&\quad-\;\mathbb{E}_{D}\left[\frac{{p}_{1t}(i,\boldsymbol{X},\boldsymbol{A})}{1 - {p}_{1g}(i,\boldsymbol{X},\boldsymbol{A})}\left(\Delta Y_i - \Delta {m}_{0t}(i,\boldsymbol{X},\boldsymbol{A})\right)\mid D=0,G=g\right]p_{0g} \\[10pt]
&= \mathbb{E}_{D}\left[\left(\Delta m_{1g}(\boldsymbol{X},\boldsymbol{A}) - \Delta m_{0g}(\boldsymbol{X},\boldsymbol{A})\right)\mid D=1,G=g\right] \\[5pt]
&\quad-\;\mathbb{E}_{D}\left[\frac{{p}_{1g}(i,\boldsymbol{X},\boldsymbol{A})}{1 - {p}_{1g}(i,r\boldsymbol{X},\boldsymbol{A})}\left(\Delta m_{0g}(i,\boldsymbol{X},\boldsymbol{A}) - \Delta m_{0g}(i,\boldsymbol{X},\boldsymbol{A})\right)\mid D=0,G=g\right] p_{0g} \\[10pt]
&= \tau^{DATT}(g).
\end{aligned}
\]
where the third step applies the law of iterated expectations, and the final step is justified by the conditional parallel trends assumption.

\textbf{Case 2: When propensity score model is correctly speciﬁed.}
In this case, we have that 
\[
\begin{aligned}
{\tau}^{dr}(g)
&= \mathbb{E}_{D} \left[ \left( D_i \mathbf{1}\{G_i=g\} - \frac{(1 - D_i)\mathbf{1}\{G_i=g\}{\pi}_{1g}(i,\boldsymbol{X},\boldsymbol{A})}{1 - {\pi}_{1g}(i,\boldsymbol{X},\boldsymbol{A})} \right)\left( \Delta Y_i - \Delta {u}_{0g}(i,\boldsymbol{X},\boldsymbol{A}) \right) \right] \\[10pt]
&= \mathbb{E}_{D} \left( \left( D_i \mathbf{1}\{G_i=g\} - \frac{(1 - D_i)\mathbf{1}\{G_i=g\}\pi_{1g}(i,\boldsymbol{X},\boldsymbol{A})}{1 - \pi_{1g}(i,\boldsymbol{X},\boldsymbol{A})} \right)\Delta Y_i  \right) \\[10pt]
&\quad - \mathbb{E}_{D} \left[ \left( D_i \mathbf{1}\{G_i=g\} - \frac{(1 - D_i)\mathbf{1}\{G_i=g\} \pi_{1g}(i,\boldsymbol{X},\boldsymbol{A})}{1 - \pi_{1g}(i,\boldsymbol{X},\boldsymbol{A})} \right) \Delta {u}_{0g}(i,\boldsymbol{X},\boldsymbol{A}) \right]\\
&= \tau^{DATT}(g) - \mathbb{E}_{D}\mathbb{E}\left[(\pi_{1g}-\pi_{1g})\Delta {u}_{0g}(i, \boldsymbol{X}, \boldsymbol{A})\right]\\
&= \tau^{DATT}(g).
\end{aligned}
\]

The third equality follows from Lemma 3.1 in Abadie (2005)\cite{abadie2005semiparametric} and the law of iterated expectations, reducing exactly to the formulation in their paper when the indicator \( G\) is omitted.

\end{proof}

\subsection*{Proof of  Theorem 1:}
Before proving Theorem 1, we first introduce a definition and a lemma. 
\begin{definition}
    A triangular array \(\{Z_i\}_{i=1}^n\) is conditionally \(\psi\)-dependent given \(\mathcal{F}_n\) if there exists a constant \(C > 0\) and an \(\mathcal{F}_n\)-measurable sequence \(\{\psi_n(s)\}_{s,n\in\mathbb{N}}\) with \(\psi_n(0) = 1\) for all \(n\) such that for every \(n, h, h' \in \mathbb{N}\), every \(s > 0\), every function \(f \in \mathcal{L}_h\) and \(f' \in \mathcal{L}_{h'}\), and every pair \((H, H') \in \mathcal{P}_n(h, h'; s)\), we have
\[
\Big|\operatorname{Cov}\Big(f(\boldsymbol{Z}_H),\, f'(\boldsymbol{Z}_{H'})\Big)\Big|
\le C\, h\, h'\,\Big(\|f\|_{\infty} + \operatorname{Lip}(f)\Big)
\Big(\|f'\|_{\infty} + \operatorname{Lip}(f')\Big)
\psi_n(s)
\]
almost surely; here, \(\psi_n(s)\) is called the dependence coefficient of the array.
\end{definition}

\begin{lemma}
    Under Assumptions \ref{assumption4}, \ref{assumption5}, \ref{assumption6}(a), \ref{assumption6}(b) hold, then for any \(g\in \mathcal{G}\), the sequence \(\{\phi_{i}(g)\}_{i=1}^n\) is conditionally \(\psi\)-dependent given \((\boldsymbol{X}, \boldsymbol{A})\) as per Definition 1, with the dependence coefficient \(\psi_n(s)\) defined by (\ref{34}).
\end{lemma}
\begin{proof}
Let \(\mathcal F_n\) be the \(\sigma\)-algebra generated by \((\boldsymbol X,\boldsymbol A)\), \((h,h')\in\mathbb N^2\), \((f,f')\in\mathcal L_h\times\mathcal L_{h'}\), \(s>0\), and \((H,H')\in\mathcal P_n(h,h';s)\).  
Fix \(g\in\mathcal G\) and write
\[
\phi_i(g)\;=\;\underbrace{\Big(D_i\,\mathbf 1\{G_i=g\}-\frac{(1-D_i)\,\mathbf 1\{G_i=g\}\,p_{1g}(i,\boldsymbol X,\boldsymbol A)}{1-p_{1g}(i,\boldsymbol X,\boldsymbol A)}\Big)}_{=:W_i(g)}
\cdot\Big(\Delta Y_i-\Delta u_{0g}(i,\boldsymbol X,\boldsymbol A)\Big)\;-\;\tau^{DATT}(g).
\]
Define \(Z_i=\phi_i(g)\), \(\boldsymbol Z_H=(Z_i)_{i\in H}\), \(\xi=f(\boldsymbol Z_H)\), and similarly \(\zeta=f'(\boldsymbol Z_{H'})\).

For fix s, take \(D_{jt}^{(s/2)} = l_{jt}\big(\boldsymbol X^{\mathcal N(j,s/2)},\boldsymbol A^{\mathcal N(j,s/2)},\boldsymbol\nu^{\mathcal N(j,s/2)}\big)\) and $\boldsymbol{D}_{\mathcal{N}\left(i, {s^{\prime}}/2\right)}^{(s/2)}=\left(D_j^{(s/2)}\right)_{j \in \mathcal{N}\left(i, {s^{\prime}}/2\right)}$, define the $s/2$-local exposure indicator \(\mathbf 1_i^{(s/2)}(g)=\mathbf 1\{G(i,\boldsymbol D_{\mathcal{N}\left(i, {s}/2\right)}^{(s/2)},\boldsymbol A_{\mathcal{N}\left(i, {s}/2\right)})=g\}\), the $s/2$-local difference \(\Delta Y_i^{(s/2)} =\Delta h_{it}(\boldsymbol D_{\mathcal{N}\left(i, {s}/2\right)}^{(s/2)},\boldsymbol X^{\mathcal N(i,s/2)},\boldsymbol A^{\mathcal N(i,s/2)},\boldsymbol\epsilon^{\mathcal N(i,s/2)} \), and the $s/2$-local weight
\[
W_i^{(s/2)}(g)\;=\;D_i^{(s/2)}\,\mathbf 1_i^{(s/2)}(g)\;-\;\frac{(1-D_i^{(s/2)})\,\mathbf 1_i^{(s/2)}(g)\,p_{1g}(i,\boldsymbol X,\boldsymbol A)}{1-p_{1g}(i,\boldsymbol X,\boldsymbol A)}.
\]
Set
\[
Z_i^{(s/2)} \;=\; W_i^{(s/2)}(g)\,\Big(\Delta Y_i^{(s/2)}-\Delta u_{0g}(i,\boldsymbol X,\boldsymbol A)\Big)\;-\;\tau^{DATT}(g).
\]

Hence,
 \(\big(Z_i^{(s/2)}\big)_{i\in H}\perp\!\!\!\perp \big(Z_j^{(s/2)}\big)_{j\in H'}\mid\mathcal F_n\), then we have
\[
\begin{aligned}
\big|\mathrm{Cov}(\xi,\zeta\mid \mathcal F_n)\big|
&\le \big|\mathrm{Cov}(\xi-\xi^{(s/2)},\zeta\mid\mathcal F_n)\big|+\big|\mathrm{Cov}(\xi^{(s/2)},\zeta-\zeta^{(s/2)}\mid\mathcal F_n)\big|\\
&\le 2\|f'\|_\infty\,\mathbf E\!\left[\,|\xi-\xi^{(s/2)}|\ \big|\ \mathcal F_n\right]
+2\|f\|_\infty\,\mathbf E\!\left[\,|\zeta-\zeta^{(s/2)}|\ \big|\ \mathcal F_n\right]\\
&\le 2\!\left(h\|f'\|_\infty\mathrm{Lip}(f)+h'\|f\|_\infty\mathrm{Lip}(f')\right)
\max_{i\in\mathcal N_n}\mathbf E\!\left[\,|Z_i-Z_i^{(s/2)}|\ \big|\ \mathcal F_n\right].
\end{aligned}
\]
Thus it remains to bound \(\max_i \mathbf E[|Z_i-Z_i'|\mid \mathcal F_n]\). Write
\[
Z_i-Z_i^{(s/2)} \;=\; \underbrace{\big(W_i(g)-W_i^{(s/2)}(g)\big)}_{\text{weight gap}}
\cdot\big(\Delta Y_i-\Delta u_{0g}\big)\;+\;
W_i^{(s/2)}(g)\cdot\underbrace{\big(\Delta Y_i-\Delta Y_i^{(s/2)}\big)}_{\text{outcome gap}}.
\]
Hence, for some constant \(C_0>0\),
\[
\mathbf E\!\left[\,|Z_i-Z_i^{(s/2)}|\ \big|\ \mathcal F_n\right]
\;\le\; C_0\Big(
\mathbf E\!\left[\,|W_i(g)-W_i^{(s/2)}(g)|\ \big|\ \mathcal F_n\right]
+ \mathbf E\!\left[\,|\Delta Y_i-\Delta Y_i^{(s/2)}|\ \big|\ \mathcal F_n\right]\Big).
\]
Under Lemma 2,
\[
\mathbf E\!\left[\,|W_i(g)-W_i^{(s/2)}(g)|\ \big|\ \mathcal F_n\right]
\;\le\; C_{1}\Big(\eta_n(s/2)+n(i,K)\eta_n(s/2)\Big).
\]
Under Lemma 3,
\[
\left|\mathbf E\!\left[\,|\Delta Y_i-\Delta Y_i^{(s/2)}|\ \big|\ \mathcal F_n\right]\right|
\;\le\; 2\gamma_n(s/2)+\Lambda_n(i,s/2)\,n(i,s/2)\,\eta_n(s/2).
\]
Then there exists \(C_2>0\) such that
\[
\max_{i\in\mathcal N_n}\mathbf E\!\left[\,|Z_i-Z_i'|\ \big|\ \mathcal F_n\right]
\;\le\;
C_2\cdot
\underbrace{\max_{i\in\mathcal N_n}\Big(\gamma_n(s/2)+\eta_n(s/2)\,[1+n(i,K)+\Lambda_n(i,s/2)\,n(i,s/2)]\Big)}_{=:~\psi_n(s)}.
\]
\end{proof}

\begin{assumption}[\textbf{Local Lipschitz Continuity}]\label{assumption9}
For each $t\in\{1,2\}$ there exists $\Lambda_n(i,s)>0$ such that for all $\mathbf d,\mathbf d'\in\{0,1\}^n$,
    \[
    \big|h_{it}(\mathbf d^{\mathcal{N}(i,s)},\mathbf X,\mathbf A,\boldsymbol\varepsilon_t^{\mathcal{N}(i,s)})
    -h_{it}(\mathbf d'^{\mathcal{N}(i,s)},\mathbf X,\mathbf A,\boldsymbol\varepsilon_t^{\mathcal{N}(i,s)})\big|
    \ \le\ \Lambda_n(i,s)\sum_{j\in {\mathcal{N}(i,s)}}|d_j-d'_j|.
    \]

\end{assumption}

\begin{lemma}
Fix $s$, and abbreviate
$
D'_j=l_{jt}\!\left(\boldsymbol X^{\mathcal N(j,s)},\boldsymbol A^{\mathcal N(j,s)},\boldsymbol\nu^{\mathcal N(j,s)}\right),
\quad
\boldsymbol D'_B=(D'_j)_{j\in B},\ B\subseteq\mathcal N_n,
$
and define, for any exposure value $g\in\mathcal G$,
\[
\mathbf 1_i(g)=\mathbf 1\!\left\{G\!\big(i,\boldsymbol D^{\mathcal N(i,K)},\boldsymbol A\big)=g\right\},\qquad
\mathbf 1_i'(g)=\mathbf 1\!\left\{G\!\big(i,\boldsymbol D'^{\mathcal N(i,K)},\boldsymbol A\big)=g\right\}.
\]
Under Assumption \ref{assumption1} and \ref{assumption4} hold,
\[
\mathbf E\!\left[\,\big|\,\mathbf 1_i(g)-\mathbf 1_i'(g)\,\big| \ \Big| \ \boldsymbol X,\boldsymbol A\right]
\ \le\ n(i,K)\,\eta_n(s).
\]

\end{lemma}

\begin{proof}
Let $J:=\mathcal N(i,K)$ and enumerate $J=\{j_1,\dots,j_m\}$ with $m=|J|$. Define a sequence of treatment vectors by changing one coordinate in 
$J$ at a time:
\[
\boldsymbol D^{(0)}=\boldsymbol D,\qquad
\boldsymbol D^{(r)}=\text{ same as }\boldsymbol D^{(r-1)}\text{ except }(\boldsymbol D^{(r)})_{j_r}=D'_{j_r},\ r=1,\dots,m.
\]
Since $G$ is $K$-local, hence
\[
\begin{aligned}
\big|\mathbf 1_i(g)-\mathbf 1_i'(g)\big|
&\le \sum_{r=1}^m 
\Big|\mathbf 1\!\big\{G(i,\boldsymbol D^{(r)},\boldsymbol A)=g\big\}
-\mathbf 1\!\big\{G(i,\boldsymbol D^{(r-1)},\boldsymbol A)=g\big\}\Big| \\
&\le \sum_{r=1}^m |D_{j_r}-D'_{j_r}|
= \sum_{j\in J} |D_j-D'_j|.
\end{aligned}
\]
Taking conditional expectations and using Assumption \ref{assumption4},
\[
\mathbf E\!\left[\,\big|\,\mathbf 1_i(g)-\mathbf 1_i'(g)\,\big| \ \Big| \ \boldsymbol X,\boldsymbol A\right]
\ \le\ \sum_{j\in J}\mathbf E\!\left[\,|D_j-D'_j|\,\big|\,\boldsymbol X,\boldsymbol A\right]
\ \le\ n(i,K)\,\eta_n(s).
\]

\end{proof}

\begin{lemma}
Let \(B_i=\mathcal N(i,s)\), \(n(i,s)=|B_i|\). Define: 
\[
D'_{jt}=l_{jt}\!\big(\mathbf X^{\mathcal N(j,s)},\mathbf A^{\mathcal N(j,s)},\boldsymbol\nu_t^{\mathcal N(j,s)}\big),\quad
\mathbf D'^{\,B_i}_t=(D'_{jt})_{j\in B_i},
\]
and 
\[
Y'_{it}=h_{it}\!\big(\mathbf D'^{\,B_i}_t,\mathbf X^{B_i},\mathbf A^{B_i},\boldsymbol\varepsilon_t^{B_i}\big).
\]
Let \(\Delta Y_i=Y_{i2}-Y_{i1}\) and \(\Delta Y_i'=Y'_{i2}-Y'_{i1}\).
Under Assumptions \ref{assumption4}, and \ref{assumption9},
\[
\Big|\ \mathbb E[\Delta Y_i\mid \mathbf X,\mathbf A]\ -\ \mathbb E[\Delta Y_i'\mid \mathbf X,\mathbf A]\ \Big|
\ \le\ 2\gamma_n(s)+ \Lambda_n(i,s)\,n(i,s)\,\eta_n(s).
\]
\end{lemma}

\begin{proof}
    By Assumption \ref{assumption4} and the tower property,
    \[
\Big|\mathbb E[Y_{it}\mid \mathbf X,\mathbf A]
-\mathbb E\!\big[h_{it}(\mathbf D_t^{B_i},\mathbf X^{B_i},\mathbf A^{B_i},\boldsymbol\varepsilon_t^{B_i})\mid \mathbf X,\mathbf A\big]\Big|
\le \gamma_n(s).
\]
Subtracting \(t=1\) from \(t=2\) and applying the triangle inequality,
\[
\Big|\ \mathbb E[\Delta Y_i\mid \mathbf X,\mathbf A]
-\mathbb E\!\big[\Delta h_i(\mathbf D^{B_i}_t)\mid \mathbf X,\mathbf A\big]\ \Big|
\le 2\gamma_n(s),
\]
Then using Assumption \ref{assumption2}, \ref{assumption4} and \ref{assumption9},
\[
\begin{aligned}
\Big|\mathbb E[\Delta h_i(\mathbf D^{B_i}_t)-\Delta h_i(\mathbf D'^{\,B_i}_t)\mid \mathbf X,\mathbf A]\Big| 
&= \Big|\mathbb E[h_{i2}(\mathbf D_2^{B_i},\cdot)-h_{i2}(\mathbf D'^{\,B_i}_2,\cdot)\mid \mathbf X,\mathbf A]\Big|\\
&\le \mathbb E\!\left[\ \big|h_{i2}(\mathbf D_2^{B_i},\cdot)-h_{i2}(\mathbf D'^{\,B_i}_2,\cdot)\big|\ \Big|\ \mathbf X,\mathbf A\right]\\
&\le \Lambda_n(i,s)\sum_{j\in B_i}\mathbb E\!\left[\,|D_{j2}-D'_{j2}|\,\Big|\,\mathbf X,\mathbf A\right]\\
&\le \Lambda_n(i,s)\,n(i,s)\,\eta_n(s),
\end{aligned}
\]
Therefore,
\[
\Big|\ \mathbb E[\Delta Y_i\mid \mathbf X,\mathbf A]-\mathbb E[\Delta Y_i'\mid \mathbf X,\mathbf A]\ \Big|
\le 2\gamma_n(s)+\Lambda_n(i,s)\,n(i,s)\,\eta_n(s).
\]
\end{proof}
Now we are ready to prove Theorem 1. 
\begin{proof}
We start with the difference
$
\sqrt{m_n}\bigl(\hat{\tau}^{dr}(g) - \tau^{DATT}(g)\bigr).
$
The first step is to write this difference as a main sum plus a few remainder terms. Specifically,

\[
\sqrt{m_n} \bigl(\hat{\tau}^{dr}(g) - \tau^{DATT}(g)\bigr)
\;=\;
\underbrace{\frac{1}{\sqrt{m_n}} \sum_{i \in \mathcal{M}_n} \phi_{i}(g)}_{\text{main term}}
\;+\; R_{1} \;+\;  R_{2} \; ,
\]
where
\begin{align*}
R_{1} 
\; &=\;
\frac{1}{\sqrt{m_n}} 
\sum_{i \in \mathcal{M}_n} 
(1 - D_i)\, \mathbf{1}(G_i = g)\, \left(\Delta Y_i - \Delta u_{0g}(i, \mathbf{X}, \mathbf{A})\right)
\frac{\hat{p}_{1g}(i, \mathbf{X}, \mathbf{A}) - p_{1g}(i, \mathbf{X}, \mathbf{A})}
     {(1 - \hat{p}_{1g}(i, \mathbf{X}, \mathbf{A}))(1 - p_{1g}(i, \mathbf{X}, \mathbf{A}))}, \\[1.5ex]
R_{2} 
\; &=\;
\frac{1}{\sqrt{m_n}}
\sum_{i\in\mathcal M_n}
\frac{\hat p_{1g}(i,\mathbf X,\mathbf A)-D_i\mathbf{1}(G_i=g)}{1-\hat p_{1g}(i,\mathbf X,\mathbf A)}
(\Delta u_{0g}(i,\mathbf X,\mathbf A)-\Delta\hat u_{0g}(i,\mathbf X,\mathbf A)).
\end{align*}
The function \(\phi_{i}(g)\) captures the leading contribution of unit \(i\) to the difference between the estimand and true targets. To establish the asymptotic properties of the main term, we introduce the concept of \(\psi\)-dependence as defined in Kojevnikov (2021)\cite{kojevnikov2021limit} to characterize weak dependence. Let \(\mathcal{L}_{d, t}\) represent the set of all real-valued functions \( f(\cdot) \) defined on \(\mathbb{R}^{\nu \times h}\) that are bounded and Lipschitz continuous. $\operatorname{Lip}(f) \text { be the Lipschitz constant of } f \in \mathcal{L}_{d, t}$
Additionally, define the collection of subset pairs as:
\[
\mathcal{P}_M(h, h'; s) = \Big\{(H, H') \,\Big|\, H, H' \subseteq D_M, |H| = h, |H'| = h', \ell_\textbf{A}(H, H') \geq s \Big\}.
\]
This set \(\mathcal{P}_M(h, h'; s)\) consists of all pairs \((H, H')\) of subsets drawn from \(D_M\), where
- \( H \) and \( H' \) have sizes \( h \) and \( h' \), respectively. The minimum separation distance \(\rho(H, H')\) between the two subsets is at least \( s \), ensuring a certain level of weak dependence between them.

Given sigma-algebra \(\mathcal{F}_n\) generate by \((\mathbf{X}, \mathbf{A})\) , the collection \(\{\phi_{i}(g)\}_{i=1}^n\) is \(\psi\)-dependent. By assumptions on boundedness and dependence (Assumptions \ref{assumption7}(a), \ref{assumption7}(b)), one can apply central limit theorem for \(\psi\)-dependent sequences (Kojevnikov et al., 2021\cite{kojevnikov2021limit}, Theorem 3.2). We can obtain 
\[
\sigma_n^{-1} \;
\frac{1}{\sqrt{m_n}} 
\sum_{i \in \mathcal{M}_n} 
\phi_{i}(g)
\;\;\xrightarrow{d}\;\; 
\mathbf{N}(0,1),
\]
i.e., normalized by \(\sigma_n\), the main term converges in distribution to standard normal.  

Thus, it remains to show that the remainder terms \(R_{1}\) and \(R_{2}\) are negligible. We begin by writing \(\mu_{0g}(i)=\mu_{0g}(i,\boldsymbol{X},\boldsymbol{A})\), \({p}_{1g}(i)={p}_{1g}(i,\boldsymbol{X},\boldsymbol{A})\) and \(\hat{p}_{1g}(i)=\hat{p}_{1g}(i,\boldsymbol{X},\boldsymbol{A})\). Squaring and taking expectations yields
\begin{align*}
\mathbf{E}[(R_{1})^2] 
&= \frac{1}{m_n} \sum_{i,j\in\mathcal{M}_n} \mathbf{E} \left[ \mathbf{E} \Bigl[ (\Delta Y_{i} - \Delta\mu_{0g}(i))(\Delta Y_{j} - \Delta \mu_{0g}(j)) \mid \boldsymbol{D},\boldsymbol{X},\boldsymbol{A} \Bigr] \right. \\
&\quad \left. \times \frac{\mathbf{1}_i(g)\mathbf{1}_j(g)(\hat{p}_{1g}(i) - p_{1g}(i))(\hat{p}_{1g}(j) - p_{1g}(j))}{(1-\hat{p}_{1g}(i))(1- p_{1g}(i))(1- \hat{p}_{1g}(j)) (1-p_{1g}(j))} \right] \\
&\leq \frac{CC'}{m_n} \sum_{i,j\in\mathcal{M}_n} \gamma_n\Bigl(\frac{\ell_{\boldsymbol{A}}(i,j)}{2}\Bigr)^{1-2/p} \mathbf{E} \left[ |\hat{p}_{1g}(i) - p_{1g}(i)| |\hat{p}_{1g}(j) - p_{1g}(j)| \right] \\
&\quad \text{(using Assumption \ref{assumption6} and Lemma C.5 in \cite{leung2024graph})} \\
&= \frac{CC'}{m_n} \sum_{s=0}^{\infty} \gamma_n(s/2)^{1-2/p} \sum_{i,j\in\mathcal{M}_n} \mathbf{1} \{ \ell_{\boldsymbol{A}}(i,j) = s \} \\
&\quad \times \mathbf{E} \left[|\hat{p}_{1t}(i) - p_{1t}(i)| |\hat{p}_{1t}(j) - p_{1t}(j)| \right] \\
&\quad \text{(grouping pairs by network distance)} \\
&\leq \frac{CC'}{m_n} \sum_{s=0}^{\infty} \gamma_n(s/2)^{1-2/p} \left( \sum_{i,j} \mathbf{1} \{ \ell_{\boldsymbol{A}}(i,j) = s \} \right)^{1/2} \\
&\quad \times \left( \sum_{i,j} \mathbf{1} \{ \ell_{\boldsymbol{A}}(i,j) = s \} \mathbf{E} \left[ (\hat{p}_{1t}(i) - p_{1t}(i))^2 \right] \right)^{1/2} \\
&\leq CC' \sum_{s=0}^{\infty} \gamma_n(s/2)^{1-2/p} \frac{n}{m_n} \left( \frac{1}{n} \sum_{i=1}^n |\mathcal{N}^{\partial}(i,s)|^2 \right)^{1/2} \\
&\quad \times \left( \frac{1}{n} \sum_{i=1}^n \mathbf{E} \left[ (\hat{p}_{1t}(i) - p_{1t}(i))^2 \right] \right)^{1/2}.
\end{align*}

Under Assumptions \ref{assumption6} and \ref{assumption7}(b), the terms on the right-hand side converge to zero, implying that \(\mathbf{E}[R_{1}]=o_p(1)\) and thus \(R_{1}\) is negligible. Then, following the proof of Theorem 3.1 in Farrell (2021)\cite{farrell2021deep}, we obtain that \( R_{2}  = o_p(1) \).
Because each of the remainder terms \(R_{1}, R_{2}\) is shown to be negligible relative to the main term, they do not affect the limiting distribution. This establishes that
\[
\sqrt{m_n}\bigl(\hat{\tau}^{dr}(g) - \tau^{DATT}(g)\bigr)
\;\;\xrightarrow{d}\;\;
\mathbf{N}\bigl(0,\,\sigma^2\bigr),
\]
for some limit variance \(\sigma^2\).

\end{proof}

\subsection*{Proof of  Theorem 2:}

Define  
\[
\hat{\phi}_{i}(g) = \hat{\tau}^{dr}_i(g) - \hat{\tau}^{dr}(g),
\]  
and let
\[
\hat{\sigma}^2 = \frac{1}{m_n} \sum_{i \in \mathcal{M}_n}\sum_{j \in \mathcal{M}_n}\hat{\phi}_{i}(g)\hat{\phi}_{j}(g)\mathbf{1}\{\ell_{\boldsymbol{A}}(i,j)\leqslant b_n\},
\]
\[
{\sigma}^2 = \frac{1}{m_n} \sum_{i \in \mathcal{M}_n}\sum_{j \in \mathcal{M}_n}\phi_{i}(g)\phi_{j}(g)\mathbf{1}\{\ell_{\boldsymbol{A}}(i,j)\leqslant b_n\}.
\]

We first aim to show the convergence result:
\[
|\hat{\sigma}^2 - {\sigma}^2|\xrightarrow{p}0.
\]

Note that:
\[
\begin{aligned}
|\hat{\sigma}^2 - {\sigma}^2|
&= \left|\frac{1}{m_n}\sum_{i \in \mathcal{M}_n}\sum_{j \in \mathcal{M}_n}\left(\hat{\phi}_{i}(g)\hat{\phi}_{j}(g) - \phi_{i}(g)\phi_{j}(g)\right)\mathbf{1}\{l_{\boldsymbol{A}}(i,j)\leqslant b_n\}\right|\\[6pt]
&= \left|\frac{1}{m_n}\sum_{i \in \mathcal{M}_n}\left(\hat{\phi}_{i}(g)-\phi_{i}(g)\right)\sum_{j \in \mathcal{M}_n}\left(\hat{\phi}_{j}(g)+\phi_{j}(g)\right)\mathbf{1}\{l_{\boldsymbol{A}}(i,j)\leqslant b_n\}\right| \\[6pt]
&\leqslant \frac{n}{m_n}\left(\frac{1}{n}\sum_{i=1}^n\left(\hat{\phi}_{i}(g)-\phi_{i}(g)\right)^2\right)^{1/2}\left(\frac{1}{n}\sum_{i=1}^n\max_{j\in\mathcal{N}_n}\left(\hat{\phi}_{j}(g)+\phi_{j}(g)\right)^2n(i,b_n)^2\right)^{1/2}.
\end{aligned}
\]

Next, using Theorem \ref{theorem1}, we can show that
\[
\frac{1}{n} \sum_{i=1}^n\left(\hat{\phi}_{i}(g)-\phi_{i}(g)\right)^2=o_p\left(n^{-1 / 2}\right).
\]

And by Assumptions \ref{assumption5} and \ref{assumption8}(c), we have, for some universal constant \( C > 0 \), that

\[
\frac{1}{n}\sum_{i=1}^{n}\max_{j\in\mathcal{N}_n}\left(\hat{\phi}_{j}(g)+\phi_{j}(g)\right)^2 n(i,b_n)^2 
\leqslant 
C\frac{1}{n}\sum_{i=1}^{n} n(i,b_n)^2
=
O_p(\sqrt{n}).
\]

Then, we have $|\hat{\sigma}^2 - {\sigma}^2|$ is $o_p(1)$.

Next, following the proof strategy in Theorem 4 of Leung (2022)\cite{leung2022causal}, we can establish that

\[
{\sigma}^2 = \hat{\sigma}_*^2 + R_n + o_p(1).
\]

Specifically, this result follows by adapting Leung's (2022)\cite{leung2022causal} arguments, where we replace his term \( Z_i - \tau_i(t) \) with our \({\phi}_{i}^*(g)\), and utilize our Assumptions \ref{assumption8}(b)–(d) in place of his Assumptions 7(b)–(d).
Finally, by applying Proposition 4.1 from Kojevnikov et al. (2021)\cite{kojevnikov2021limit}, we can establish that \(\left|\hat{\sigma}_*^2 - \sigma_n^2\right| \xrightarrow{p} 0\).

\section*{Appendix B: Results for Spillover Effects}
Beyond estimating the direct average treatment effect on the treated, empirical researchers may also seek to evaluate the spillover average treatment effect on the treated, which is defined as

\begin{equation}
 \tau^{SATT}(g ; d)
 =
 \frac{1}{m_n} \sum_{i \in \mathcal{M}_n}\mathbb{E}
 \left[
 Y_{i2}(d, g)
 -
 Y_{i2}(d, 0)
 \mid
 D_i=1, G_i=g, \boldsymbol{X}, \boldsymbol{A}
\right].
\end{equation}

Identification is relatively straightforward for \(  \tau^{SATT}(g ; 1) \), since the potential outcomes \( Y_{i2}(1, g) \) for units who receive treatment under exposure level \( g \) are directly observed in the data. However, the corresponding counterfactual outcomes \( Y_{i2}(0, g) \) for these same individuals—i.e., what their outcomes would have been under control, given the same exposure—are not observed. To identify the direct effect of treatment assignment at each exposure level \( g \), we impose a parallel trends assumption on \( Y_{i2}(0, g) \), as follows:

\begin{assumption}[\textbf{Network Conditional Parallel Trends for \(  \tau^{SATT}(g ; 1) \)}]

\begin{align}
& \mathbb{E} \left[ Y_{i2}(0, g) \mid D_i = 1, G_i = g, \boldsymbol{X}, \boldsymbol{A} \right] - \mathbb{E} \left[ Y_{i1}(0, g) \mid D_i = 1, G_i = g, \boldsymbol{X}, \boldsymbol{A} \right]\notag\\
= \; & \mathbb{E} \left[ Y_{i2}(0, g) \mid D_i = 0, G_i = g, \boldsymbol{X}, \boldsymbol{A} \right] - \mathbb{E} \left[ Y_{i1}(0, g) \mid D_i = 0, G_i = g, \boldsymbol{X}, \boldsymbol{A} \right].
\end{align}

\end{assumption}
Consistent with the direct effects framework, the following expressions represent the doubly robust estimands for the spillover effects:

\[
\delta^{dr}(g,1) = \frac{1}{m_n} \sum_{i \in \mathcal{M}_n} \left[\left(D_i \mathbf{1}\{G_i = g\} - \frac{ D_i \mathbf{1}\{G_i = 0\} {p}_{1g}(i, \boldsymbol{X}, \boldsymbol{A})}{ {p}_{10}(i, \boldsymbol{X}, \boldsymbol{A})} \right) \left( \Delta Y_i - \Delta {u}_{10}(i, \boldsymbol{X}, \boldsymbol{A}) \right) \right].
\]

Following the same logic, a doubly robust estimator for $\delta^{dr}(g,0)$ can be constructed. The asymptotic distribution of the spillover ATT effect estimators can be established similarly, following almost the same approach as that for the direct ATT effect.

\section*{Appendix C: Results for Multiple Time Periods with Staggered Treatment}
When treatment timing is common across units, extending the framework to multiple time periods is straightforward. We simply aggregate all pre-treatment periods into one and all post-treatment periods into another, denoting them as \( t = 1 \) and \( t = 2 \), respectively. In contrast, when treatment is staggered across units, the situation becomes more complex. If we are interested in the ATT effect at a specific time \( t \), the conventional approach is to compare units that receive treatment at time \( t \) with those that never receive treatment, as in Callaway and Sant’Anna (2021)\cite{callaway2021difference}. The main limitation of this method is that units already treated before time \( t \) may affect the potential outcomes of those treated at time \( t \), thereby compromising the identification of the causal effect.

We consider a standard staggered DID setting with four groups and four periods. In each period, one additional group begins treatment, and once treated, a group remains treated. Only three groups receive treatment, so one group never receives treatment throughout. This structure is illustrated in Figure \ref{figure2}.

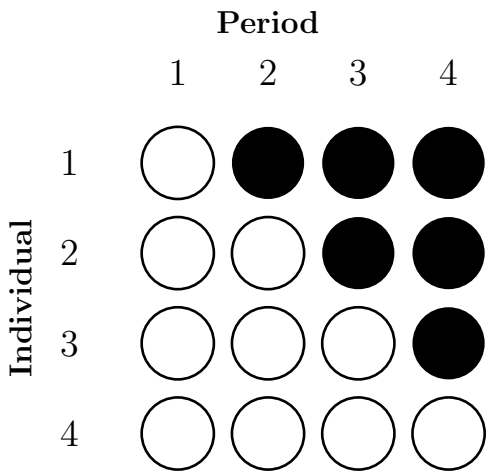
\begin{figure}
\centering
\caption{Common Staggered DID design}
\label{fig:staggered_did}

\begin{tikzpicture}[x=1.2cm,y=-1.2cm]
% 坐标标签
\node[font=\bfseries, anchor=north] at (2.5,-0.3) {Period};
\node[font=\bfseries, rotate=90, anchor=south] at (0,3) {Individual};

% 刻度数字
\foreach \col in {1,...,4}
  \node[font=\large] at (\col+0.5,0.5) {\col};

\foreach \row in {1,...,4}
  \node[font=\large] at (0.3,\row+0.5) {\row};

% 圆点绘制
\foreach \row in {1,...,4}{%
  \foreach \col in {1,...,4}{%
    \ifthenelse{\row=1 \AND \col>1}
      {\fill (\col+0.5,\row+0.5) circle (0.4);}
    {\ifthenelse{\row=2 \AND \col>2}
      {\fill (\col+0.5,\row+0.5) circle (0.4);}
    {\ifthenelse{\row=3 \AND \col>3}
      {\fill (\col+0.5,\row+0.5) circle (0.4);}
      {\draw[line width=1pt] (\col+0.5,\row+0.5) circle (0.4);}
    }}}
}

\end{tikzpicture}
\label{figure2}
\end{figure}

%------------------- 图二：Network Staggered DID -------------------%
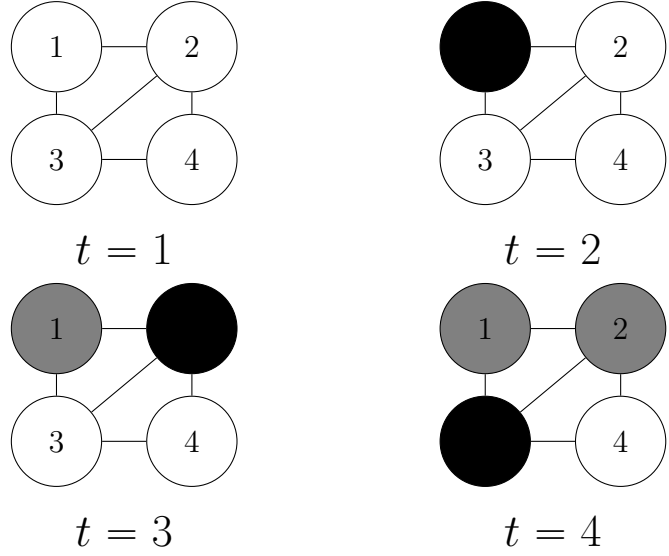
\begin{figure}
\centering
\caption{Network Staggered DID }
\label{fig:network_did}

\begin{tikzpicture}[scale=1.5, every node/.style={circle, draw, minimum size=1.2cm}]

% ---------- t = 1 ----------
\begin{scope}[xshift=0cm, yshift=2.5cm]
  \node (1) at (0,1) {1};
  \node (2) at (1.2,1) {2};
  \node (3) at (0,0) {3};
  \node (4) at (1.2,0) {4};
  \draw (1) -- (2); \draw (1) -- (3); \draw (2) -- (3);
  \draw (2) -- (4); \draw (3) -- (4);
  \node[draw=none] at (0.6,-0.8) {\Large $t=1$};
\end{scope}

% ---------- t = 2 ----------
\begin{scope}[xshift=3.8cm, yshift=2.5cm]
  \node[fill=black] (1) at (0,1) {1};
  \node (2) at (1.2,1) {2};
  \node (3) at (0,0) {3};
  \node (4) at (1.2,0) {4};
  \draw (1) -- (2); \draw (1) -- (3); \draw (2) -- (3);
  \draw (2) -- (4); \draw (3) -- (4);
  \node[draw=none] at (0.6,-0.8) {\Large $t=2$};
\end{scope}

% ---------- t = 3 ----------
\begin{scope}[xshift=0cm, yshift=0cm]
  \node[fill=gray] (1) at (0,1) {1};
  \node[fill=black] (2) at (1.2,1) {2};
  \node (3) at (0,0) {3};
  \node (4) at (1.2,0) {4};
  \draw (1) -- (2); \draw (1) -- (3); \draw (2) -- (3);
  \draw (2) -- (4); \draw (3) -- (4);
  \node[draw=none] at (0.6,-0.8) {\Large $t=3$};
\end{scope}

% ---------- t = 4 ----------
\begin{scope}[xshift=3.8cm, yshift=0cm]
  \node[fill=gray] (1) at (0,1) {1};
  \node[fill=gray] (2) at (1.2,1) {2};
  \node[fill=black] (3) at (0,0) {3};
  \node (4) at (1.2,0) {4};
  \draw (1) -- (2); \draw (1) -- (3); \draw (2) -- (3);
  \draw (2) -- (4); \draw (3) -- (4);
  \node[draw=none] at (0.6,-0.8) {\Large $t=4$};
\end{scope}

\end{tikzpicture}
\label{figure3}
\end{figure}

When consider a staggered DID design in a networked setting, where treatment propagates not only through direct assignment but also via neighboring exposure. Let there be four groups of individuals, each indexed by \( i = 1, \ldots, 4 \). These groups may be connected to one another through a known undirected network structure, represented by the adjacency matrix \( A \), as illustrated in Figure \ref{figure3}. In this setting, the black nodes indicate the group that receives treatment in the current period, gray nodes represent groups that have already been treated in previous periods, and white nodes correspond to groups that have not yet received any treatment. For illustration, We consider the exposure mapping $T_i$ as

\[
T_i = \left( D_i,\ \sum_{j=1}^n A_{ij} D_j \right).
\]

Under this design, the parallel trends  assumption boils down to comparing treated and untreated groups that share the same number of treated neighbors. At \(t = 2\), the newly treated Group 1 has zero treated neighbors; among the three still-untreated groups, only Group 4 likewise has no treated neighbors, making Groups 1 and 4 the valid comparison pair. The same logic carries over to \(t = 3\) and \(t = 4\): when Groups 2 and 3 receive treatment, Group 4 remains the only group with an identical count of treated neighbors, so it continues to serve as the appropriate control for the treated groups in those periods.

In the simple example above, we merely wanted to show that when spillover effects are present, a staggered DID design must isolate units that satisfy our conditional parallel trends  assumption to achieve causal identification. Refining the search for units that meet this assumption is the price one pays for pinning down more specific causal effects. For instance, identifying the direct average treatment effect on the treated.

\section*{Appendix D: Results for Repeated Cross-Sectional Data}
In a repeated cross-sectional design, we obtain two independent samples, one drawn before and one after the intervention.  The period-0 observations \((Y_0,D_0,X_0,A_0)\) follow a joint distribution \(P_0\), while the period-1 observations \((Y_1,D_1,X_1,A_1)\) follow a different distribution \(P_1\).  We allow arbitrary dependence among units within each period—for example, network-induced correlations—but assume the two periods are independent of one another.

For each unit \(i\), every exposure value \(g\in\mathcal G\) generated by the exposure mapping \(G_i\), treatment status \(d\in\{0,1\}\), and survey wave \(t\in\{0,1\}\), define the network-specific outcome regression: $ \mu_{dg}(i,t,\mathbf X,\mathbf A)
      =\mathbb{E}\!\left[
        Y_i \mid
        T_i=t,\;
        D_i=d,\;
        G_i=g,\;
        \mathbf X,\;
        \mathbf A
      \right]
$, and $
\Delta\mu_{dg}\!\bigl(i,\mathbf X,\mathbf A\bigr)
   = \mu_{dg}\!\bigl(i,1,\mathbf X,\mathbf A\bigr)
   - \mu_{dg}\!\bigl(i,0,\mathbf X,\mathbf A\bigr)
$, and $
\mu^{\,rc}_{dg,Y}\!\bigl(i,T,\mathbf X,\mathbf A\bigr)
   \;=\;
   T\,\mu_{dg}\!\bigl(i,1,\mathbf X,\mathbf A\bigr)
   +\bigl(1-T\bigr)\,\mu_{dg}\!\bigl(i,0,\mathbf X,\mathbf A\bigr).
$

These definitions parallel the Sant’Anna-and-Zhao notation while explicitly conditioning on the full covariate matrix \(\mathbf X\) and adjacency matrix \(\mathbf A\), thereby allowing the mean outcome to vary flexibly with each unit’s position in the network. Then for the case in which repeated cross-section data are available, we consider the estimator:
\[
{
\hat{\tau}^{drrc}{}(g)
   =\frac1{n}\sum_{i=1}^n
     \Bigl[\,\Delta w_{1,g}(D_i,T_i,G_i)
            -\Delta w_{0,g}\!\bigl(D_i,T_i,\mathbf X_i;\hat p_{1g}\bigr)\Bigr]
     \bigl[Y_i-\hat\mu^{rc}_{0,Y}\,(i,T_i,\mathbf X_i,\mathbf A)\bigr]},
\]
where the period-specific weights equal
\[
\begin{aligned}
w^{rc}_{1,t,g}(D,T,G)
      &={D\,\mathbf 1\{T=t\}\,\mathbf 1\{G=g\}}
            ,\\[6pt]
w^{rc}_{0,t,g}\!\bigl(D,T,G;\hat p_{1g}\bigr)
      &=\frac{\hat p_{1g}(X,\mathbf A)\,(1-D)\,\mathbf 1\{T=t\}\,\mathbf 1\{G=g\}}
             {1-\hat p_{1g}(X,\mathbf A)}.
\end{aligned}
\]
Calculating the difference between the first period and the pre-period we have
\[
\begin{aligned}
\Delta w^{rc}_{1,g}(D,T,G)&=w^{rc}_{1,1,g}(D,T,G)-w^{rc}_{1,0,g}(D,T,G),\\
\Delta w^{rc}_{0,g}\!\bigl(D,T,G;\hat p_{1g}\bigr)
      &=w^{rc}_{0,1,g}-w^{rc}_{0,0,g}.
\end{aligned}
\]
Under the regularity conditions introduced for the panel results, namely (i) the network-dependent law of large numbers and central-limit theorem of Kojevnikov et al. (2021)\cite{kojevnikov2021limit}; (ii) mean-square consistency of at least one first-step GNN learner; and (iii) the network-conditional parallel trends assumption, the repeated-cross-section estimator \(\widehat{\tau}^{rc}(g)\) is \(\sqrt n\)-consistent and asymptotically normal. The proof mirrors the argument used for the panel estimator, see Sant’Anna and Zhao (2020)\cite{sant2020doubly} for the same derivation in the i.i.d. repeated-cross-section case.

\section*{Appendix E: Supplementary Simulation Results}

The simulation in the main text primarily focuses on estimating $\tau^{DATT}$. Since our main objective is to demonstrate how applying GNNs to estimate the nuisance function can effectively mitigate bias introduced by confounder network interference, in this part we simulate different model configurations, specifically emphasizing the handling of treatment network interference. We intentionally exclude network confounding via covariates in order to isolate and better understand the effects of treatment interference alone. 

Specifically, our simulation model is defined as follows. We first generate the network adjacency matrix \(\boldsymbol{A}  \) using the same procedure as described in the main paper. The covariates \(X_1\) and \(X_2\) are independently drawn from a standard normal distribution. A nonlinear combination is then applied: $X = 1 + \frac{X_2}{1 + \exp(X_1)}$. All random error terms used throughout the simulation are independently drawn from \( \mathbf{N}(0, 1) \) as well.

We simulate a binary treatment indicator \( D_i \in \{0,1\} \) through a logistic model. Specifically, we define:
\[
V_i = \theta_{D,1} 
+ \theta_{D,2} \cdot X_i 
+  \nu_i,
\]
and define the probability of treatment as \( \pi_i = \frac{1}{1 + \exp(-V_i)} \). The treatment variable is then drawn as \( D_i \sim \text{Bernoulli}(\pi_i) \). The parameter vector is specified as: $
\theta_D = (0.4,\, 1.5).
$

Then we define outcomes for the pre-treatment and post-treatment periods as follows
\[
Y_{\text{pre},i} =
\theta_{\text{pre},1} 
+\theta_{\text{pre},2} \cdot D_i+ \theta_{\text{pre},3} \cdot D_i \cdot \mathbf{1} \left\{ \sum_{j=1}^n A_{ij} D_j > 0 \right\}
+ \theta_{\text{pre},4} \cdot X_i
+  \epsilon_i,
\]
with parameters: $
\theta_{y,\text{pre}} = (1,\, 0,\, 0,\, 0.6),
$
and
\[
Y_{\text{post},i} =
\theta_{\text{post},1} 
+ \theta_{\text{pre},2} \cdot D_i+\theta_{\text{post},3} \cdot D_i \cdot \mathbf{1} \left\{ \sum_{j=1}^n A_{ij} D_j > 0 \right\}
+ \theta_{\text{post},4} \cdot X_i
+  \mu_i,
\]
with parameters $
\theta_{y,\text{post}} = (0.5,\, 0.2,\,0.2,\, 0.8).
$
In this setup, we use \(\mathbf{1}\left\{ \sum_{j=1}^n A_{ij} D_j > 0 \right\}\) as the exposure mapping. Under the above model specifications, the true exposure-specific DATTs are \(\tau^{\text{DATT}}(0) = 0.2\) and \(\tau^{\text{DATT}}(1) = 0.4\).

\begin{table}[htbp]
\centering
\caption{Simulation Results for treatment spillover}
\resizebox{\textwidth}{!}{
\begin{tabular}{cccccccccc}
\toprule

$n$ &  & 500 &  &  & 1000 &  &  & 2000 &  \\
\# treated         &  & 196 &  &  & 391 &  &  & 792 &  \\
\midrule

$\hat{\tau}^{DATT}(1)$ &  & 0.39974 &  &  & 0.39978 &  &  & 0.39981 &  \\
  & & (0.07094) & &  & (0.04971) &  &  & (0.03490) &  \\
$\hat{\tau}^{DATT}(0)$ & & 0.20314 &  &  & 0.20155 &  &  & 0.20084 &  \\
 &  & (0.06466) &  &  & (0.04582) &  &  & (0.03246) &  \\

\midrule

DR-DID &  & 0.26358 &  &  & 0.26321 &  &  & 0.26423 &  \\ 
            &  & (0.05113) &  &  & (0.03609)  & &  & (0.02551) &  \\

\bottomrule
\end{tabular}}
\label{table1}
\end{table}
We estimate \(\tau^{\text{DATT}}(1)\) and \(\tau^{\text{DATT}}(0)\) with the proposed DATT estimators which explicitly account for treatment spillover effects. We compare with the DR-DID estimator of Sant’Anna and Zhao (2020)\cite{sant2020doubly} which ignores heterogeneity in spillover exposure. Table 3 shows the estimation results, with standard errors in parenthesis. As can be seen from Table 3, the Network DR-DID method delivers accurate DATT estimates under various sample sizes. The traditional DR-DID estimates lie between \(\tau^{\text{DATT}}(1)\) and \(\tau^{\text{DATT}}(0)\). The traditional DR-DID estimator conflates treatment effects across different exposure groups. As a result, its estimate essentially averages across treated units in both exposure groups, making it difficult to interpret the true causal effects when spillovers are present. 

\end{document}